\documentclass[11pt]{article}
\pdfoutput=1
\usepackage{microtype}
\usepackage{amsmath,amssymb}
\usepackage{hyperref}
\usepackage{wrapfig}
\usepackage{amsthm}

\newcommand{\Z}{\mathbb{Z}}

 \newcommand\ForAuthors[1]
 {\par\smallskip                     
  \begin{center}
   \fbox
   {\parbox{0.9\linewidth}
    {\raggedright\sc--- #1}
   }
  \end{center}
  \par\smallskip                     %
 }        
\newtheorem{definition}{Definition}
\newtheorem{lemma}{Lemma}
\newtheorem{theorem}{Theorem}
\newtheorem{example}{Example}

\def\F{\mbox{$Fin_+$}}

\def\cG{{\cal G}}
\def\cI{{\cal I}}

\def\cO{{\cal O}}
\def\cP{{\cal P}}

\def\cR{{\cal R}}
\def\cS{{\cal S}}
\def\cT{{\cal T}}

\def\Gz{\cG_0}
\def\Gf{\cG_f}
\def\Gp{\cG_P}

\usepackage[disable,colorinlistoftodos,bordercolor=orange,backgroundcolor=orange!20,linecolor=orange,textsize=scriptsize]{todonotes}

\newcommand{\I}[1]{\mathit{#1}}
\DeclareMathOperator{\Lit}{Lit}
\newcommand{\la}{\langle}
\newcommand{\ra}{\rangle}
\newcommand{\Lcal}{\mathcal{L}}
\newcommand{\Msf}{\textsf{M}}
\newcommand{\Ssf}{\textsf{S}}
\newcommand{\ssf}{\textsf{s}}
\newcommand{\pre}{\textsf{pre}}
\newcommand{\RM}[1]{{\rm{#1}}}
\newcommand{\tsf}{\textsf{t}}

\usepackage{tikz}  
\usetikzlibrary{positioning}
\usetikzlibrary{decorations.pathmorphing,shapes}

\newcommand{\act}{{\sf ACT}}
\newcommand{\acte}{{\sf act_e}}
\newcommand{\sched}{{\sf Sched}}
\newcommand{\ba}{{\bf a}}
\newcommand{\jact}{\mbox{$\overline{\act}$}}
\newcommand{\sfa}{{\sf a}}
\newcommand{\comment}[1]{}

\newcommand{\Conz}{{\sf Con}_0}

\newcommand{\rd}{{\sf read}}
\newcommand{\wrt}{{\sf write}}
\usepackage{hyperref}
\newcommand{\footremember}[2]{%
    \footnote{#2}
    \newcounter{#1}
    \setcounter{#1}{\value{footnote}}%
}

\title{A simplicial complex model of dynamic epistemic logic for
 fault-tolerant distributed computing}
\author{Eric Goubault\footremember{X}{LIX, Ecole Polytechnique, CNRS, Universit\'e Paris-Saclay, 91128 Palaiseau, France \texttt{goubault@lix.polytechnique.fr}} \and 
Sergio Rajsbaum\footremember{UNAM}{Instituto de Matematicas, UNAM, Ciudad Universitaria Mexico 04510, Mexico \texttt{rajsbaum@im.unam.mx}}} 

\begin{document}
\maketitle
The usual epistemic S5 model for multi-agent systems is a Kripke graph, whose edges are labeled with the agents that do not distinguish between two states. We propose  to uncover the higher dimensional information implicit in the Kripke graph,  by using as a model its dual, a chromatic simplicial complex. For each state of the Kripke model there is a facet in the complex, with one vertex per agent. If an edge $(u,v)$ is labeled with a set of agents $S$, the facets corresponding to $u$ and $v$  intersect in a simplex consisting of one vertex for each  agent of $S$. Then we use dynamic epistemic logic  to study how the simplicial complex epistemic model changes after the agents communicate with each other. We show that there are topological invariants preserved from the initial epistemic complex to the epistemic complex after an action model is applied, that depend on how reliable  the communication is. In turn these topological properties determine the knowledge that the agents may gain after the communication happens.

\tableofcontents

\newpage

\section{Introduction}


Dynamic epistemic logic (DEL) considers multi-agents systems and studies 
how their knowledge changes when communication events occur.
It extends ordinary epistemic logic by the inclusion of \emph{event models} to describe actions.
An epistemic Kripke model represents the knowledge of the agents about an initial situation,
and an event model represents their knowledge about the possible event taking place in this situation.
  The most common situation studied, which gave birth to DEL,
   is that of  a public announcement to all the agents of a formula $\psi$,
   but there is a
   general logical language to reason about information and knowledge change~\cite{baltag&:98,BaltagM2004}.
A product update operator  defines a Kripke model that results as a 
  consequence of executing actions on the initial Kripke model.

\vskip -1cm

\paragraph*{Overview}
 An epistemic  S5 model is typically used to represent  states of a multi-agent system, where edges 
 of the Kripke structure are labeled
with the agents that do not distinguish between the two states.
In this work we show there are  underlying topological invariants induced by the action model,
that determine what the agents may know after a communication event takes place.

We argue that while the Kripke model seems to be a one-dimensional structure (a graph), it  actually encodes
a high dimensional topological object, namely, a \emph{simplicial complex} corresponding,
in a precise categorical sense, to the dual of the Kripke structure. Thus, there is a sort of epistemic complex
representing the knowledge of the agents about the initial situation, and another epistemic complex,
which represents their knowledge obtained through the product with the communication action model. 

We show that these  complexes indeed carry topological information. In the figure below, $I$ represents an initial
epistemic simplicial complex model (equivalent to Kripke model), and $P$ the product with an action model.
The action model preserves topological invariants from the initial model $I$  to the complex $P$ after
the communication actions have taken place. Actually, how much of the topology is preserved, depends
on how reliable is the communication among the agents (see explanation below),
and encode the
degree at which some knowledge does not evolve.

We explore a class of action models that \emph{fully} preserve the topology of the initial complex,
and hence, in a precise sense, yield the \emph{least} information to the agents.
\begin{wrapfigure}{r}{-0.4\textwidth}
 \centering
\begin{tikzpicture}
  \node (s) {$P$};
  \node (xy) [below=2 of s] {$\Delta \subseteq I \times O$};
  \node (x) [left=of xy] {$I$};
  \draw[<-] (x) to node [sloped, above] {$\pi_I$} (s);
  \draw[->, dashed, left] (s) to node {$h$} (xy);
  \draw[->] (xy) to node [below] {$\pi_I$} (x);
\end{tikzpicture}
\end{wrapfigure}
We  define another \emph{knowledge goal} action model, that when it is used
to make the product with the initial epistemic model, yields an epistemic model $\Delta$
representing what the agents should be able to know, after applying the communication action model.
There is sufficient knowledge in $P$ if there exists a (properly defined)  morphism $h$ from $P$ to  $\Delta$ that makes the diagram commute.

We stress that this is far from telling the whole story. Our formalization of knowledge goals
represents only the case of what the processes should learn \emph{once}. Indeed, we provide
a precise categorical equivalence between dynamic epistemic Kripke models of tasks 
and tasks as  epistemic complexes (so we have two versions of this diagram, one for Kripke models and one for complex models).
A more general
situation would study requirements on how knowledge should evolve with time.
Also, we concentrate on a very specific setting. The class of communication action models we study, represents all possible ways in
which the agents can send to each other messages (but we use shared memory to model the communication pattern), allowing for any pattern
of arrivals, including message losses. Roughly speaking, assume
 a set of agents $A=\{0,\ldots,n\}$ exchange messages once.   Each agent sends 
 a single message to every other agent. A communication event
 defines which messages are delivered and in which order, assuming no partitions happen. 
 Although our setting could be used to study even weaker situations, where messages are lost
 in a way that there is no communication between two groups of agents.
 This is, essentially, the \emph{one round} communication action model we consider, and
 where we can show that the topology is fully preserved, and knowledge gained is minimal (except
 for of course  weaker partitioning situations). Then we study multi-round versions of the action model,
 and show that also in these the topological invariants are fully preserved.
 
Finally, let us mention that what we developped here is oriented towards distributed
tasks, where problem are specified in terms of input-output (values) relationships. 
This does not account for inputless problems such as counting problems, timestamps
problems etc., for which our framework is going to be generalized in a subsequent paper. 

\vskip -1cm

\paragraph*{Main results}
As a concrete case study of the theory described above, we consider a multi-agent system
representing a distributed computing system, where processes can fail by crashing, and
communicate by reading and writing shared variables. Using shared memory instead of message passing
simplifies
some of the technical development.
Indeed,
we have been inspired to consider simplical complexes as the dual of Kripke models
from the study of distributed computability of tasks using topology~\cite{HerlihyKR:2013}.
This approach to analysing distributed computations 
represents executions as simplicial complexes. As the computation
evolves, the topology of the complex changes. A \emph{task} is defined by a relation from an input complex $\cI$ specifying
the input values to the processes, and an output complex $\cO$ specifying values that  the processes should decide
after communicating with each other. The seminal result in the area is the algebraic topology  \emph{set agreement} impossibility, saying
that $n+1$ processes cannot agree in at most $n$ different
values, \emph{wait-free}, i.e., if the processes are asynchronous and may crash~\cite{1993GeneralizedFLPImposibility_BG,1999TopologicalStructureAsynchronous_HS,SaksZ00}.
The categorical equivalence between  Kripke models and simplicial complex epistemic models we establish provides,
in one direction,  a DEL semantics to distributed computability,
 opening the possibility of reasoning  about  knowledge change in distributed
computing. For instance, the set agreement impossibilty is exposed as the impossibility of gaining sufficient knowledge 
as required by a decision action model.
In the other direction, the connection allows to bring in the topological invariants known in distributed computing,
to DEL, and in particular show that knowledge gained after an epistemic action model
is intimately related to  higher dimensional topological properties.

Section~\ref{subsec:model}  describes the distributed computing model used as a case study.
 We formalize task solvability for distributed computing in Section \ref{sec:kripkeModel}
in terms of Kripke frames, one of the classical models for multi-modal S5 logics. 
For this Kripke frame formalization, we note that the classical ``carrier map'' approach
of \cite{HerlihyKR:2013} can be expressed as the existence of a certain 
commutative diagram involving only morphisms of Kripke frames, see Theorem \ref{thm:Kripketasksolv}. We then show 
in Section~\ref{sec:dynEpSemantics} that these Kripke frames can be used to model a logic of
knowledge for agents, and lift Theorem \ref{thm:Kripketasksolv} on Kripke models, this
is Theorem \ref{thm:Kripketasksolv2}. There is then a simple epistemic interpretation
of task solvability~: the protocol can only improve knowledge of the agents, and a task
is solvable only if it has improved so that to have at least the knowledge formalized
by the specification. The other major contribution of this paper is Section \ref{sec:categor}. We first show that the category of (proper) Kripke frames is equivalent to the category
of pure chromatic simplicial complexes, Theorem \ref{thm:equiv}, 
which is the one in which many results
have been developped using combinatorial topology, over the years, see \cite{HerlihyKR:2013}. This also lifts to Kripke models and what we call simplicial models, Theorem
\ref{thm:equiv2}. From this, we show that the classical protocol complex approach
of \cite{HerlihyKR:2013} and our epistemic logic approach are equivalent. The interpretation in terms of knowledge progress and loss of it, has its topological counterpart~: there are some topological invariants on the simplicial models that measure this, extending the intuition that common knowledge is linked to some connectivity
conditions. 
Conclusions and open problems are in  Section~\ref{sec:conclusions}.

Note that we could also have directly defined an interpretation of epistemic logics in (pure chromatic) and later, proved its correctness. We prefered to base our presentation on classical grounds (Kripke models) and show and use an equivalence
of categories to end up with a correct interpretation of epistemic logics in simplicial complexes, hence in a geometric framework. 

\vskip -.5cm
\paragraph*{Related work}
Epistemic logic has been used very successfully and many times to analyze and design fault tolerant distributed systems.
In this paper  we use 
DEL~\cite{sep-dynamic-epistemic,DEL:2007}.
Complex epistemic actions can be
represented in \emph{action model logic} \cite{BaltagM2004,DEL:2007}. Various examples of epistemic actions have been considered, especially 
\emph{public announcement logic},  a well-studied example of DEL, with many applications in dynamic logics, knowledge representation and other
formal methods areas. However, to the best of our knowledge, it has not been used in distributed computing theory, where fault-tolerance
is of primal interest. 
DEL~\cite{BaltagM2004,DEL:2007}
extends epistemic logic through dynamic operators formalizing
information change.
Plaza~\cite{plaza:89} first extended epistemic logic to model public
announcements, where the same information is transmitted to all
agents.
Next, a variety of approaches (e.g.,~\cite{baltag&:98,DEL:2007}) generalized
such a logic to include communication that does not necessarily reach
all agents.
Here, we build upon the approach developed by Baltag et
al.~\cite{baltag&:98} employing ``action models''.
We have focused in this paper on the classical semantics of multi-modal S5 logics. 
The seminal work on the
subject, see e.g. \cite{sep-dynamic-epistemic}, has considered topological models. 
Future work will 
include the relationship between these topological models and 
the geometric realization of our simplicial models of Section~\ref{sec:categor}.
Another classical model for epistemic logic is the one
of Interpreted Systems of  Fagin, Halpern, Moses and Vardi (see e.g. \cite{InterpretedSystems}). These are related to the semantics we
give of our distributed systems, in Section \ref{sec:states}. More importantly, the adjunctions between Kripke frames and Interpreted Systems, developped
in \cite{Porter} are in the same spirit as the equivalence of categories between (proper) Kripke frames and pure chromatic
simplicial complexes that we develop in Section \ref{equivalencesimpcomp}. The structure of pure chromatic
simplicial complexes has the advantage to behave well with respect to products 
with a natural interpretation in terms of action models, 
and also conveys some important geometrical contents, that
we develop in Section \ref{equivalencesimpcomp}. 

Seminal work on knowledge and distributed systems is of course one of the inspirations
of the present work, see e.g. \cite{FHMVbook}, as well as the  combinatorial topology approach for fault-tolerant distributed computing,
see e.g. \cite{HerlihyKR:2013}. Quite a few results have been achieved using 
epistemic logic to model 
fault-tolerant distributed protocols, see e.g. 
\cite{MosesR2002} but the authors know no previous work on relating the combinatorial
topological methods of \cite{HerlihyKR:2013} with Kripke models.
We should note that the formulation of carrier maps as products which is
developed in Section \ref{sec:carrierprod} seems to have 
been partially observed in 
\cite{Havlicek2000}. 

\section{Distributed computing models}

\label{subsec:model}
In a  \emph{wait-free} model of $n+1$ processes,  a process does not use instructions
that wait for events in other processes. Wait-free computation has been studied thoroughly e.g.~\cite{Herlihy:waitFree1988,Raynal-waitFree2005}, and
has turned out to be fundamental in the theory of distributed computing~\cite{HerlihyKR:2013}. Simulations and reductions can be used
to transfer results to other models e.g.~\cite{BorowskyGLR01,GafniKM:generalizedACT:2014,HerlihyR12}. For instance, 
a wait-free   
model can be used as a basis to study models where  $t$ processes may crash (and it makes sense for
a process to wait until it hears from $n+1-t$ processes). Wait-free models tolerate any number of processes crashing,
yet the treatment of failures is substantially simplified, and are mostly taken care implicitly, as we shall see below. 
Indeed, we consider only models where failures are not detectable in a finite execution (because it
is indistinguishable if a  process  crashed or is just slow); as opposed to synchronous models, where if a message
does not arrive by a given time, a process has crashed.
We concentrate in models where  communication is by read/write shared variables, but the framework can be used to
study more powerful communication primitives available in practice e.g.~\cite{Kogan:2012:MCF}, as well as message passing models.

\subsection{Informal overview of fault-tolerant distributed computing theory}
We start by recalling some of the basics of distributing computing, most relevant to this work; several textbooks contain
further details~\cite{2004dc_AW,HSbook,LynchBook:1996}.

\subsubsection{Overview of distributed systems}
\label{sec:infModels}
The most basic model we consider is the   {\it one-round read/write} model ($\mathsf{WR}$), e.g.~\cite{2004dc_AW}. It consists of $n+1$ processes  denoted by the numbers $[n]=\{0,1,\ldots,n\}$. A process is a deterministic (possibly infinite) state machine. Processes  communicate through a shared memory array $\mathsf{mem}[0\ldots n]$ which consists of $n+1$ single-writer/multi-reader atomic registers. Each process  accesses the shared memory by invoking the atomic operations $\mathsf{write}(x)$ or $\mathsf{read}(j)$, $0\leq j\leq n$. The $\mathsf{write}(x)$ operation is used by process $i$  to write  value $x$ to its own register, $i$, and process $i$  can invoke $\mathsf{read}(j)$ to read register $\mathsf{mem}[j]$, for any $0\leq j\leq n$. 
 In its first operation, process $i$ writes a value  to $\mathsf{mem}[i]$, then it reads each of the $n+1$ registers, in an arbitrary order. 
 Such a sequence of operations, consisting of a write followed by all the reads, is abbreviated by  $\mathsf{WScan}(x)$.

A \emph{state}  of the system consist of the state of each process and the state of the shared memory.
An \emph{execution} is defined by an initial state of the system  and  an interleaving of the operations of the processes; it is
an alternating sequence of states and operations.
In the most basic case, \emph{every} interleaving of the operations is a possible execution. Namely, we have a fully asynchronous system.

Several models  derived from the one-round read/write $\mathsf{WR}$ model have been considered in the literature.
\begin{itemize}
\item
 In the \emph{multi-round} version of the $\mathsf{WR}$ model,
  processes can  execute any number of write and read operations on the array $\mathsf{mem}$. 
  It is often convenient to assume the program of a process is structured in rounds, each one consisting
  of a $\mathsf{WScan}()$ operation\footnote{This assumption is done without loss of generality
  for computability purposes; to reduce complexity a program would not always be structured this way.}.
  This is the most commonly considered version, e.g.~\cite{2004dc_AW}.

\item
In the \emph{iterated} $\mathsf{WR}$ model, processes communicate through a sequence of arrays. They all go through the sequence of arrays $\mathsf{mem}_0$, $\mathsf{mem}_1\ldots$ in the same order, executing a
single $\mathsf{WScan}()$ operation on $\mathsf{mem}_r$, for each $r\geq 0$.
Namely,  process $i$ executes one write to $\mathsf{mem}_r[i]$  and then reads one by one all entries $j$, $\mathsf{mem}_r[j]$, in arbitrary order, before proceeding to do the same on $\mathsf{mem}_{r+1}$. For an overview of iterated models see e.g.~\cite{RajsbaumIterated2010}.

Notice that the one round  $\mathsf{WR}$ model is the special case of either the multi-round or the iterated 
$\mathsf{WR}$ models, when the program of all processes consists of a single  $\mathsf{WScan}()$ operation.

 \item
 The \emph{snapshot} versions of the previous, multi-round and iterated models, is obtained by replacing
 the $\mathsf{WScan}()$ by a $\mathsf{WSnap}()$ operation, that guarantees that the reads 
 of  the $n+1$ registers happen all atomically, at the same time. In this case an execution is an interleaving
 of $\mathsf{write}$ and $\mathsf{snap}$ operations by the processes. Snapshot operations can be implemented
 wait-free and are a very useful abstraction, see e.g.~\cite{AADGMSsnaps}.

  \item
 The \emph{immediate snapshot} version of the previous multi-round and iterated models is obtained by assuming
 that the whole $\mathsf{WSnap}()$ operation happens ``atomically," and hence is called
 an  $\mathsf{ImmSnap}()$ operation. More precisely,
 an execution in such a model is represented by a sequence of concurrency classes,
 where in each class a set of  $\mathsf{ImmSnap}()$ operations take place concurrently~\cite{AttiyaR2002,1993GeneralizedFLPImposibility_BG,SaksZ00}.
Thus, in a {concurrency class}, the  writes  by the set of processes participating in the concurrency class 
occur in arbitrary order, followed by a read to all registers by each of these processes, in arbitrary order. 
The iterated version has been considered since~\cite{1997SimpleAlgorithmicallyReasoned_BG}.
\end{itemize}

\subsubsection{Overview of distributed tasks}
A  \emph{decision task} is the distributed equivalent of a function, where each process knows only part of the input, and after communicating with the other processes, each process computes part of the output. For instance, in the $k$-\emph{set agreement} task~\cite{1993MoreChoices_Ch} each process starts with an input value from some set of input values $V$ with $|V| >k $, and each process
has to decide an output value, that is one of the input values, and such that
 at most $k$ different values are decided; when $k=1$ we get the  \emph{consensus} task. The $k$-set agreement
 task is solvable (by asynchronous processes communicating using read/write registers) if and only if at most $t$ processes can crash, for $t<k$, see~\cite{1993GeneralizedFLPImposibility_BG,BorowskyGLR01,FischerLP85,1999TopologicalStructureAsynchronous_HS,SaksZ00}. 

\subsubsection{Computability of distributed tasks}
A central concern in distributed computability is studying which tasks are solvable in a given distributed computing model (as determined by e.g. the type of communication mechanism available  and the reliability of the processes). Early on it was shown that consensus is not solvable even if only one process can fail by crashing,
when the processes are  asynchronous and they communicate by message passing~\cite{FischerLP85}, or even by writing and reading a shared memory~\cite{LouiAA:87}. A graph theoretic characterization of the tasks solvable in the presence of at most one process failure appeared soon after~\cite{BiranMZ90}. 

The  \emph{asynchronous computability theorem}~\cite{1999TopologicalStructureAsynchronous_HS} exposed that moving from tolerating one process  failure, to any number of process failures, yields a characterization of the class of decision tasks that can be solved in a wait-free\footnote{When any number of processes may crash, the algorithm run by a process must be \emph{wait-free,} because there is no reason for it to wait for an event occurring in another process.} manner by asynchronous processes based on simplicial complexes, which are higher dimensional versions of graphs. In particular, $n$-set agreement is not wait-free solvable,  with  $n+1$ processes~\cite{1993GeneralizedFLPImposibility_BG,1999TopologicalStructureAsynchronous_HS,SaksZ00}.

Computability theory through combinatorial topology has evolved  to encompass non-independent process failures, arbitrary malicious failures, synchronous and partially synchronous processes, and various communication mechanisms~\cite{HerlihyKR:2013}. Still,  the original wait-free model of the asynchronous computability theorem, where crash-prone processes that communicate wait-free by writing and reading a shared memory is  fundamental. Topological techniques are derived in this model, and then extended to
other models, e.g.~\cite{2013TopologuDistributedAdversaries_HR}. Also, the question of solvability in other models (e.g. $t$ crash failures, for $1\leq t\leq n$), can in many cases be reduced to the question of wait-free solvability ($t=n$), as shown in~\cite{BorowskyGLR01} and \cite{HerlihyR12}.

More specifically, in the \emph{AS model} of~\cite{HerlihyKR:2013} each process can write its own location of the shared-memory, and it is able to  read the entire shared memory in one  atomic step, called a \emph{snapshot}. The characterization  is based on the \emph{protocol complex,} which is a geometric representation of the various possible executions of a protocol. Simpler variations of this  model have been considered. In the \emph{immediate snapshot} (IS) version~\cite{AttiyaR2002,1993GeneralizedFLPImposibility_BG,SaksZ00}, processes can execute  a combined write-snapshot operation. The \emph{iterated immediate snapshot} (IIS) model~\cite{1997SimpleAlgorithmicallyReasoned_BG} is even simpler to analyze, and can be extended (IRIS) to analyze partially synchronous models~\cite{2008IteratedRestrictedImmdediate_RRT}. Processes communicate by accessing a sequence of shared arrays, through immediate snapshot operations, one such operation in  each array. The success of the entire approach hinges on the fact that the topology of the protocol complex of a model determines critical information about the solvability of the task and, if solvable, about the complexity of solution~\cite{HoestS97}.

All these  snapshot models, AS, IS,  IIS and IRIS can solve exactly the same set of tasks. However, the protocol complexes that arise from the different  models are structurally distinct. The combinatorial topology  properties of these  complexes have been studied in detail, providing  insights for why some
tasks are solvable and others are not in a  given model.

\subsection{Formal model}
\label{sec:formalGenericModel}
We adapt  the model of \cite{MosesR2002} (in turn following  the style of \cite{FHMVbook}), for
the wait-free case.

\subsubsection{States, protocols, runs}
\label{sec:states}

There is a fixed finite set of
 processes,   $0,1,2,\ldots, n$, and
an {\em environment,}~$e$, which is
used to model aspects of the system that are not modeled as being part
of
the activity or state of the processes. We will model
the  shared memory as being part
of the environment's state. In addition, we will assume that various
nondeterministic choices such as various delays 
are actions performed by the environment.
For every  $i\in\{ e,0,1,\ldots,n\}$, we assume there is
a set~$L_i$ consisting of all possible {\em local states} for~$i$.
The set of {\em global states,} which we will simply call {\em states},
will consist of
$
\cG\ = \ L_e\times L_0\times\cdots\times L_n.
$. 
We denote by $x_i$ the local state of~$i$ in the state~$x$.

In a given setting, every~$i\in\{ e,0,1,\ldots,n\}$ is
associated with a nonempty set~$\act_i$ of possible actions.
Intuitively, these model 
shared-memory operations  or local
computations the process may perform.
The environment  is in charge of {\it scheduling} the processes.
A {\em scheduling action} is a set $\sched\subseteq \{0,\ldots,n\}$
of the processes that are scheduled to move next.
We assume the existence of a set $\acte$ describing the aspects of the
environment's actions that handle everything other than the scheduling
of processes.
Without loss of generality we will assume that an environment's action
(an element in $\act_e$) is a pair $(\sched,a)$, where $\sched$ is
a scheduling action and $a\in\acte$.
A {\em joint action} is a pair $\bar\sfa= (sa,\ba)$, where
$sa= (\sched,a)$ is in $\act_e$,  and $\ba$ is a function with
domain $\sched$
such that $\ba(i) \in \act_i$ for each $i\in \sched$. Thus, $\bar\sfa$
specifies an action for the environment (via~$\ba$) and
an action for every process that is scheduled to move.
We define the set of joint actions by $\jact$. Clearly, $\jact$ is
determined by a collection of action sets $\{\act_i\}_{i= 1,\ldots,n}$
and a set $\acte$.
Joint actions are the events that cause the global
state to change into a new state.
This is formally captured by the notion of a {\em transition function},
which is a function $\tau:\cG\times\jact\to\cG$
from global states and joint actions to global states, describing how a
joint action transforms the global state.

Processes follow a {\em deterministic protocol},  a function $P_i:L_i\to
\act_i$
specifying the action that~$i$ is ready to perform in every state
of~$L_i$.
The environment however follows a {\em nondeterministic protocol}, a function
$P_e:L_i\to 2^{\act_i}\setminus\emptyset$ specifying for every state
of~$L_e$ a
nonempty
set of actions, one of which~$e$ must perform in that state.%

A run is defined by  a  sequence of
global states and the joint actions that cause the transitions among
them.
Notice that once we fix a deterministic protocol $D=
(D_0,\ldots,D_n)$
for the processes, an action $sa= (\sched,a)$ of the
environment
uniquely determines the joint action $\bar\sfa= (sa,\ba)$ that
will be performed in a given (global) state~$x$:
the set~$\sched$ determines the processes that participate in the joint
action, and $\ba(i)= D_i(x_i)$ for each~$i\in\sched$.
Formally, we model a {\em run} over~$\cG$ and~$\act_e$ as a pair
$R= (r,\alpha)$, where~$r:N\to \cG$ is a  function from the natural
numbers
to $\cG$ defining  sequence of states of~$\cG$, and
$\alpha:N\to\act_e$ defines  a corresponding sequence of environment
actions.
The intuition will be that the joint action caused by~$\alpha(k)$
and the underlying protocol leads us from a state~$r(k)$ to a
state~$r(k+1)$.
As we will see later on, once we fix a model of computation and a
protocol~$D$ for the processes to follow, there will be additional
conditions relating the sequences $r$ and~$\alpha$.
These conditions guarantee,
for example, that the actions recorded by~$\alpha$ do indeed cause the
transitions among the corresponding states recorded by~$r$.
The state~$r(0)$ is called the {\em initial state} of the run~$R$.
We denote by $r(k)_i$ (resp.,\ $r(k)_e$) the local state of process~$i$
(resp.,\ of the environment) in $r(k)$.

An {\em execution} is a 
subinterval of a run,
starting and ending in a state, as described next.
For a run $R= (r,\alpha)$ and a pair $m\le m'$ where $m$ is finite and
$m'$ is finite or infinite, we denote by $R[m,m']$ the execution
starting
at the state $r(m)$ and ending in $r(m')$ and behaving as~$R$ does
between
them.
Formally, $R[m,m']= (\sigma,\beta)$ where $\sigma$ has domain
$[0,m'-m]$ and $\beta$ has domain $[0,m'-m-1]$, and they satisfy
$\sigma(k)= r(m+k)$ and $\beta(k)= \alpha(m+k)$
for every~$k$ in their respective domains.
Notice that, in principle, the same execution can occur in different
runs, and
for that matter even at different times.
A {\em suffix} of a run~$R$ is an execution of the form $R[m,\infty]$
for some
finite~$m$; similarly, a {\em prefix} of~$R$ is an execution of the
form $R[0,m]$.

Given an execution $R$ (possibly consisting of just one state),
let us denote by $R\odot sa$ the execution that results from
extending $R$ by having the environment perform the action~$sa$.
In models in which performing a joint action at a state results in
a unique next state (which will invariably be the case in this paper),
every run of a deterministic protocol~$D$ can be represented in the form
$x\odot sa_1\odot sa_2\odot\cdots$ where $x$ is
an initial state and $sa_i$ is an environment action,
for every integer $i\ge 0$.

A  {\em system}   $\cR$ is  a set of runs.
With respect to a system~$\cS$, a state~$y$ is said to {\em extend} the
state~$x$ if there is a finite execution in some run of~$\cS$ that
starts in~$x$ and ends in~$y$.
A run~$R$ is said to {\em contain} a state~$x$ if $x$ is one of the
global states in~$R$.
For conciseness, we will use terminology such as
{\em a state~$x$ of~$\cS$}, when we mean a state~$x$
appearing in a run of $\cS$, or {\em an initial state of $\cS$},
when we mean a state appearing as an initial state in a run of $\cS$.
By convention, $x$ extends $x$ for every state~$x$ of $\cS$.

\subsubsection{Shared memory models}
The standard asynchronous shared-memory model is  well known,
and detailed formal descriptions can be found in textbooks
such as \cite{2004dc_AW,HSbook}. We now briefly review the features of the model
that are relevant for the analysis presented in this section.

We assume the standard asynchronous shared-memory model where
$n+1$ processes, $n\geq 2$, communicate by reading and writing
to single-writer/multi-reader, shared variables, and
any number of processes can crash.
A (global) {\em state}~$x$
of the system is a tuple specifying a local state~$x_i$ for every
process~$i$,
and the state of the environment, which in this case consists of
the assignment of values to the shared variables, as well as the set of
pending shared-memory operations, and the set of pending reports
for read operations that have been recorded (the value has been read)
but not yet reported to the reading process.
The pending operations are the \rd\  and \wrt~operations that have been issued for these variables and have not
yet taken effect.

The sets $\act_i$ and $\act_e$ of the actions of the
processes and the environment are defined as follows.
A process performs an action only when it is scheduled to move.
This action is either  a \rd~of a shared variable (belonging to it or to some other process),
or a \wrt\ to one of its own variables. 
For simplicity we do not include local operations, we assume they are part of the read and write operations.
An action of the environment can have one of three forms:
(a) scheduling a process to move---resulting in the process performing
an
action,
(b) performing a pending shared-memory operation, or
(c) reporting the value read in a recorded \rd\ operation to the
reading process.

A run of a given deterministic protocol~$D$ in this model is a run
$R= (r,\alpha)$ satisfying the following:
\begin{itemize}
\item[(i)] $r(0)$ belongs to a set of possible initial states,
\item[(ii)] each process follows its protocol,
 and every pair of consecutive states are related according
to the operations that take place as scheduled by the environment.
\end{itemize}
If in addition
\begin{itemize}
\item[(iii)] the 
read and write actions issued are appropriately are eventually serviced by the environment, 
\end{itemize}
then we say that the run is  {\em admissible.}

Let $\cS(D)$ be the system consisting of the set of all admissible
runs of~$D$, where 
 each  process  is scheduled to move exactly the same  number of times, $N$. If a task is solvable, it is solvable in a finite number of operations, $N$. Alternatively, when a task is not solvable,
there is no $N$ for which it is solvable. We may also consider systems that are strictly contained in  $\cS(D)$,
where although all processes execute $N$ operations, not all interleavings are included.

Recall that  $R\odot sa$ is the execution that results from
extending $R$ by having the environment perform the action~$sa$.
In a shared memory model performing a joint action at a state results in
a unique next state and assuming the environment serves immediately every read and write operation,
every run of a deterministic protocol~$D$ can be represented in the form
$x\odot sc_1\odot sc_2\odot\cdots$ where $x$ is
an initial state and $sc_i$ is a scheduling action,
for every integer $i\ge 0$. That is, it is sufficient to state which is the set of processes that are scheduled to move
in a state, to determine the action of the environment, the action of the processes, and what will be the next state.

We continue  by defining models of distributed
systems in a more general manner, that can capture snapshot and immediate snapshot models.

\subsubsection{Models of distributed computation}
\label{sec:models}

We define a generic model of computation, that can be used in situations other than shared memory.
A {\em model of distributed computation}
is determined by sets~$L_i$, $i\in\{ e,0,1,\ldots,n\}$, of local states
for the processes and the environment, and corresponding sets of actions
$\act_i$, for every~$i\in\{ e,1,2,\ldots,n\}$,
and by a tuple $M= (\Gz,P_e,\tau,\Psi)$, where the following hold:
\begin{itemize}
\item $\Gz\subseteq\cG$ is called the set of {\em initial states}.
The identity of~$\Gz$ will depend on the type of analysis
for which the model is introduced. When we focus on a
particular problem such as consensus,
$\Gz$ is the set~$\Conz$ of initial states for consensus.
\item $P_e$ is a (nondeterministic) protocol for the environment.
\item $\tau$ is a transition function.
\item $\Psi$ is a set of runs over~$\cG$ and~$\act_e$,
such that for every pair of runs~$R$ and~$R'$ that have a suffix in
common,
$R\in \Psi$ if and only if $R'\in \Psi$.
The set~$\Psi$ is called the set of {\em admissible} runs in the model.
This is a tool for specifying fairness properties of the model. For
example, properties such as ``every message sent is eventually
delivered'' or ``every process moves infinitely often''
are enforced by allowing as admissible only runs in which these
properties hold.
The condition we have on~$\Psi$ being determined by the suffixes of its
runs ensures that admissibility depends only on the infinitary behavior
of
the run.
\end{itemize}

We say that a run~$R= (r,\alpha)$ is a {\it run of the protocol
$D= (D_1,\ldots,D_n)$ in}
$M= (\Gz,P_e,\tau,\Psi)$ when
\begin{itemize}
\item[(i)] $r(0)\in\Gz$, so that~$R$ begins in a legal initial state
according to~$M$,
\item[(ii)] $\alpha(k)\in P_e(r(k)_e)$ for all~$k$,
\item[(iii)] $r(k+1)= \tau(r(k),(\alpha(k),\ba^k))$ holds for all
$k\ge 0$,
where the domain of $\ba^k$ is the set $\sched$ in~$\alpha(k)$,
and $\ba^k_i=  D_i(r(k)_i)$ for every~$i\in\sched$,%
\footnote{Given this choice, any deviations of a process from the
protocol, as may happen in a model with malicious failures, will need to
be modeled as resulting from the environment's actions.
The behavior of faulty processes in such a case will be
controlled by the environment.}
and
\item[(iv)] $R\in\Psi$, so that~$R$ is admissible.
\end{itemize}
Condition (ii) implies that the environment's action at every state
of~$R$
is legal according to its protocol~$P_e$,
and condition (iii) states that the state transitions in~$R$ are
according
to the transition function~$\tau$, assuming that the joint action is
the one determined by the environment's action and the actions that the
protocol~$D$ specifies for the processes that are scheduled to move.
A run satisfying properties (i)--(iii) but not necessarily the
admissibility
condition (iv) is called {\em a run of~$D$ consistent with~$M$}.
It is a run in which the initial state and local transitions are
according to~$D$ and~$M$, but the admissibility conditions imposed
by~$\Psi$ are not necessarily satisfied.

The notions of models and protocols give us a way of focusing on a
special
class of systems, resulting from the execution of a given protocol in
a particular model.
We denote by $\cS_N(D,M,I)$ the system $\cR$,
where $\cR$ is the set of all executions of protocol~$D$ in the model~$M$
that start in initial states from a set~$I$, where $I\subseteq \Gz$, and  each process
executes exactly $N$ operations.
When (any of) $D,M$ or $N$ are clear from the context, we may simply write $\cS(I)$.

\subsubsection{Examples of models}
\label{sec:modelsEx}

Consider a set of  initial states $\Gz$, where the initial state of a process $i$
is a pair $(i,v)$, and   $v\in\{0,1\}$ represents an input value. 
 The environment is always in the same initial state, $\epsilon$.

\paragraph{IS and IIS models}
\label{ex:RWsystems}

Consider the {iterated} immediate snapshot $\mathsf{IIS}$ model (Section~\ref{sec:infModels}) where processes communicate through a sequence of arrays
$\mathsf{mem}_0$, $\mathsf{mem}_1\ldots$.
The protocol $D$ that the processes follow consists of executing a
single $\mathsf{WScan}()$ operation on $\mathsf{mem}_r$, for each $r\geq 0$.
Namely,  process $i$ executes one write to $\mathsf{mem}_r[i]$  and then reads one by one all entries $j$, $\mathsf{mem}_r[j]$, in arbitrary order, before proceeding to do the same on $\mathsf{mem}_{r+1}$. 

The set of all $N$-step executions~$\cR$ of protocol~$D$ in the  $\mathsf{IIS}$ model
that start in initial states~$\Gz$ can be obtained by applying schedules of the following form.
Every execution in~$\cR$  is  of the form
$x\odot sc_1\odot sc_2\odot\cdots\odot sc_R$ where $x\in\Gz$ and
 $sc_i$ is a \emph{block scheduling action}
for every integer $1\leq i\leq R$. A block action is an ordered partition of the set $A=\{ 0,\ldots, n\}$, consisting
of \emph{concurrency classes} $[s_0,\ldots,s_k]$, the $s_i$ are non-empty, disjoints subsets of $A$, whose union is $A$.
Thus, $0\leq k\leq n$. When $k=n$ the execution is fully sequential (processes take immediate snapshots one after the other), and when $k=0$ it is fully concurrent (they all execute an immediate snapshot concurrently).

To apply the block action $[s_0,\ldots,s_k]$ to initial state $x$, and obtain state $x\odot [s_0,\ldots,s_k]$, the 
environment schedules the processes in the following order, to execute their read and write operations on $\mathsf{mem}_{0}$.
 It first schedules the processes in $s_0$ to execute their
write operations, and then it schedules them to execute their read operations (the specific order among writes
is immaterial, and so is the case for the reads). Then the environment repeats the same for the processes in $s_1$,
scheduling first the writes and then the reads, and so on, for each subsequent concurrency class $s_i$.
Notice that the read and write operations of block action $sc_i$ are applied to memory $\mathsf{mem}_{i}$.

Consider the composition of the block actions, $sc_1\odot sc_2\odot\cdots\odot sc_R$.
Let $IIS_R$ denote the set of all such composition of block actions. That is, 
$$
sc_1\odot sc_2\odot\cdots\odot sc_R \in IIS_R
$$ 
if and only if each $sc_i$ is an ordered partition of $A$, of the form $[s_0,\ldots,s_k]$.
Then, any execution of~$\cR$ of protocol~$D$ in the  $\mathsf{IIS}$ model can be obtained by applying
a composed block action of $IIS_R$ to an initial state in $\Gz$.
In other words, $\cR$ can be seen as the product of $\Gz$ and $IIS_R$.

 The {non-iterated} immediate snapshot $\mathsf{IS}$ model (Section~\ref{sec:infModels}) is defined in a similar way,
 except that  processes communicate through a single array
$\mathsf{mem}$.
The protocol $D$ that the processes follow consists of executing a
 $\mathsf{WScan}()$ operations repeatedly  on $\mathsf{mem}$.
Namely,  process $i$ executes one write to $\mathsf{mem}[i]$  and then reads one by one all entries $j$, $\mathsf{mem}[j]$, in arbitrary order.

Every execution in~$\cR$  is  of the form
$x\odot sc_1\odot sc_2\odot\cdots\odot sc_R$ where $x\in\Gz$, but now each
 $sc_i$ is a {concurrency class} of processes, to be scheduled as before: the 
environment schedules the processes in the following order, to execute their read and write operations on $\mathsf{mem}$.
 It first schedules the processes in $sc_1$ to execute their
write operations, and then it schedules them to execute their read operations, on the initial state, to obtain a new state,
then repeats the same schedule with processes on  $sc_2$, and so on. 

The only condition is that at the end of the execution
all processes have been scheduled the same number of times, $N$,
when applying  the composition of the block actions, $sc_1\odot sc_2\odot\cdots\odot sc_R$.
Then,  $IS_R$ denotes the set of all such composition of block actions. That is, 
$$
sc_1\odot sc_2\odot\cdots\odot sc_R \in IS_R
$$ 
if and only if each $sc_i$ is a concurrency class of $A$, and each process appears $N$ times over all concurrency classes in the execution.

An interesting example of a schedule by the environment yields a  \emph{solo execution} for process $i$, in both the IS and IIS models.
In the IS model, it schedule all $N$ first concurrency classes with a singleton set containing $i$, and then schedules all other processes
in any order. In the IIS model, it schedules $i$ first, alone, in every one of the $N$ block actions.

\section{Distributed computability in terms of Kripke frames}
\label{sec:kripkeModel}

\label{sec:modeKripkeF}
The distributed computing modeling in Section~\ref{subsec:model} is based on executions.
Here we move to the orthogonal perspective, of considering sets of states.

\subsection{Kripke frames}

We define three Kripke frames, one for the input states, one for the protocol states after $N$ steps,
 and one for the output states. Also, two types of relations between Kripke frames: 
  morphism (that will later on correspond to simplicial map) and carrier morphism (that will correspond to carrier map).

In this paper a state of a Kripke frame consists of the local states of the processes (sometimes including the state of the environment,
but usually not), and agents correspond to processes in a distributed model. Whenever  the states of a Kripke frame
have this meaning, we use the following accessibility relation.
Two (global) states $u, v \in S$ are defined to be \emph{indistinguishable} by $a$, $u \sim_a v$, 
if and only if the state of process $a$ is the
same in $u$ and in $v$.
 Notice that $u \sim_a v$ defined this way is indeed an equivalence relation.
 
\begin{definition}[Kripke frame]
 \label{def:kripkeF}
Assume
a set $A = \{a_0, a_1, \ldots a_n\}$ of $n+1$ agents.
A \emph{Kripke frame} $M = \la S,\sim^A \ra$ consists of a finite set of \emph{states}, $S$,
and a function 
 $\sim^A$, yielding for every $a \in A$ an
equivalence \emph{$a$-accessibility relation} $\sim_a\subseteq S \times S$.
\end{definition}
Intuitively, two states in a Kripke structure are connected with an edge labeled
with agent $a$ if such states are ``indistinguishable'' by $a$.
We sometimes view $M$ as a graph whose edges are labeled with the agents that do not distinguish between the
two states of the edge (implicitly, every vertex has a self-loop labeled with all agents).
We will be interested in \emph{proper} Kripke frames (also known as frames
satisfying the identity intersection property in \cite{Porter}), where any two states are distinguishable by at least one agent.

 \begin{definition} [Morphism of Kripke frames]
\label{def:weakMorph}
Let $M=\la S, \sim^A \ra$ and $N=\la T,\sim^A \ra$ be two Kripke frames. A \emph{ morphism} of Kripke
frame  $M$ to $N$ is a function $f$ from $S$ to $T$ such that for all $u, v \in S$, for all $a \in A$, 
$u \sim_a v$ implies $f(u) \sim_a f(v)$. 
\end{definition}

These morphisms are known also as weak morphisms, see e.g. \cite{Porter}
as a more classical notion of morphism is generally the one of $p$-morphism (or bounded morphism, see e.g. \cite{Porter} again), which we will
not be using in the paper. 

We call ${\cal K}$ the category of Kripke frames, with morphisms of Kripke frames as
defined above. Note that this category enjoys many interesting properties, among
which the fact that cartesian products exist (Lemma~\ref{lem:Kprod})~: 

\begin{definition}
\label{def:Kprod}
Let $M=\la S, \sim^A \ra$ and $N=\la T,\sim^A \ra$ be two Kripke frames, and
define $\la U, \sim^A\ra$ as follows~: states $U$ are pairs $u=(s,t)$ of states
$s\in S$ and $t\in T$ and the accessibility relation is defined as
$(s,t) \sim_a (s',t')$ if and only if $s\sim_a s'$ and $t \sim_a t'$. 
We call $U$ the product Kripke frame. 
\end{definition}

Now this product is indeed the categorical product~: 

\begin{lemma}
\label{lem:Kprod}
The Kripke
frame defined in Definition \ref{def:Kprod} 
is the cartesian product, denoted by $M\times N$, 
in the categorical sense, of $M$ with $N$, coming with projections $\pi_M : M\times N
\rightarrow M$ and $\pi_N : M\times N \rightarrow N$, which are morphisms of Kripke frames.
\end{lemma}

Let $M=\la S, \sim^A \ra$ and $N=\la T,\sim^A \ra$ be two Kripke frames. A \emph{carrier morphism} of Kripke
frame  $M$ to $N$ is a relation $\Phi$   from $S$ to $T$ such that for every  $s\in S$,  $\Phi(s)\subseteq T$, and
for all $u, v \in S$, 
 there exist $u'\in f(u),v'\in f(v)$, with $u' \sim_a v'$, for each $a\in A$ for which $u\sim_a v$.
Notice that  $\Phi(s)$ is non-empty, for each $s\in S$.

It easy is to check that both  morphism and carrier morphism compose,  
in particular we use the following,
where  $\cI=\la \Gz, \sim^A \ra$,  $\cP=\la \Gp,\sim^A \ra$ and $\cO=\la \Gf,\sim^A \ra$
are  arbitrary Kripke frames.

Let $M=\la S, \sim^A \ra$ and $N=\la T,\sim^A \ra$ be two Kripke frames, 
 $f$ a  morphism,  and $\Phi$ a carrier morphism from $S$ to $T$.
 We say $f$ \emph{respects} $\Phi$ if $f(s) \in \Phi(s)$, for every  $s\in S$.
 Similarly, $\Phi$ \emph{respects} a carrier morphism $\Delta$ from $S$ to $T$ if 
  $\Phi(s)\subseteq \Delta(s)$, for every $s\in S$.
When $f$ {respects} $\Phi$ we may also say that $f$ \emph{is carried} by $\Phi$.

\subsection{Task solvability}
\label{sec:tasksolvability}

In the following definition of a task, the input states $\Gz$, are such that in every initial state, the environment is
in the same state, say $\epsilon$. Thus the inputs are encoded in the initial  states of the processes only (and
the shared memory does not contain inputs to be used by the processes), and we may as well disregard the environment
state from any initial state. Similarly,  there is a set of output states, $\Gf$, with no environment state (or equal to $\epsilon$),
that represent what each process should produce as output value in an execution.
The relation $\sim^A$  consists of the indistinguishability relations defined above.
 
\begin{definition} [Task]\label{def:taskK}
A \emph{task} is a triple $\cT=\la \Gz,\Gf,\Delta \ra$. 
The set of initial initial states $\Gz$, is such that in every initial state, the environment is
in the same state, and  $\Gf$ is a set of states, with the same environment state.
Then, $\Delta$ is a carrier morphism of Kripke
frame $I=\la \Gz, \sim^A \ra$ to Kripke frame $O=\la \Gf,\sim^A \ra$.
\end{definition}

Note that we can view a task as in the definition above as a sub-Kripke frame $Z$ of 
the product frame $\Gz \times \Gf$, by just taking the states of $\Gz \times \Gf$
that are related by carrier map $\Delta$, within $\Gz \times \Gf$, with the induced
accessibility relations. 
In fact, $\Delta$, as a carrier
map, of some state $s \in \Gz$ 
is just the image by the second projection $\pi_{\Gf}$ of the sub-frame of $\Gz$ 
generated by states $g$ such that $\pi_{\Gz}(g)=s$. 
In the sequel we identify $\Delta$ with $Z$, for simplicity's
sake. 

  Consider a {model of distributed computation}
 $M= (\Gz,P_e,\tau,\Psi)$,  with initial states $\Gz\subseteq\cG$,  a deterministic protocol $D$, and fix an integer $N$.
 Recall that  $\cS_N(D,M, \Gz)$  contains the $N$-admissible executions, where each process executes the same number
 of actions.

The $N$-\emph{protocol Kripke frame}  $P=\la S, \sim^A \ra$ consists of  states $S$, where  $s\in S$, is obtained by removing the environment state
from the last state in every execution in this system $\cS_N(D,M, \Gz)$.

\begin{lemma}
\label{lem:compC}
Let $I=\la \Gz, \sim^A \ra$ be the input Kripke frame and $P=\la S, \sim^A \ra$ be the $N$ step protocol Kripke frame.
Define $\cP$ to be the relation from $\Gz$ to $S$ sending each state $s\in \Gz$ to the states at the end of executions
with initial state $s$. Then $\cP$ is a carrier morphism from $I$ to $P$.
\end{lemma}

 \begin{definition} [Solving a Task]
\label{def:solvingTaskK}
The protocol $D$ with carrier morphism $\cP$
solves task  $\cT=\la \Gz,\Gf,\Delta \ra$ in $N$ steps if
there exists  a  morphism, $f$ from $P=\la \Gp, \sim^A \ra$ to  $O=\la \Gf,\sim^A \ra$ such that the composition of  $\cP$
and $f$ is a carrier morphism
respecting $\Delta$.
\end{definition}

Let us discuss this definition.
First, assume the protocol $D$ with carrier morphism $\cP$
solves task  $\cT=\la \Gz,\Gf,\Delta \ra$ in $N$ steps 
with  a  morphism, $f$ from $P=\la Gp, \sim^A \ra$ to  $O=\la Gf,\sim^A \ra$. 
If $f$ send a state $s$ to an output state $f(s)$ then indeed each process $a$ can decide (operationally, in its program)
the value that corresponds
to its local state in $f(s)$: (i) the value is a function only of its local state, and (ii) it is consistent in the sense that
if $a$ has the same local state in two states $u, v$ then the decision of $a$ are required to be the same by $f$.

Second, these decisions are respecting the task specification, because if we consider an initial state $s_0$, then any execution
starting in $s_0$ ends in a state in $s\in \cP(s_0)$, which is then mapped to a state $t$ by $f(s)$, with $t\in \Delta(s_0)$.

Finally, if such a carrier morphism  $f$ does not exist, then it is impossible to solve the task in $N$ steps by a deterministic
protocol $D$ in model $M$, because any such protocol would actually be defining a carrier morphism, because the
decision of a protocol are based only on the local states of the processes after $N$ steps.

\subsection{Carrier maps as products}
\label{sec:carrierprod}

 Fix a {model of distributed computation}
 $M= (\Gz,P_e,\tau,\Psi)$,  with initial states $\Gz\subseteq\cG$, and a deterministic protocol $D$, with some integer $N$.
 The set  $\cS_N(D,M, \Gz)$  of $N$-admissible executions, where each process executes the same number
 of actions can be seen as a product of $\Gz$ and a set of composed scheduling actions $R$ 
where each composed
 block action is of the form $sc_1\odot sc_2\odot\cdots\odot sc_R$.
If $IIS_R$ denotes the set  composition of block actions,
$
sc_1\odot sc_2\odot\cdots\odot sc_R \in IIS_R
$ 
 we may define the product 
$
S= \Gz \times  IIS_R$
to be equal to the set of states at the end of all executions starting in a state in $\Gz$ after applying the actions of
an element in $IIS_R$. Given that this is a set $S$ of (global) states, we have the induced indistinguishability relation $\sim^A$,
and a Kripke frame, $P=\la Gp, \sim^A \ra$. 
This is of course
 the protocol Kripke frame 
 which are the states at the
end of all $N$-admissible executions starting in $\Gz$.

More generally, for model $M= (\Gz,P_e,\tau,\Psi)$,   protocol $D$ and $N$, we can consider the set
of all possible composed scheduling actions of the environment $AC$, each one determined by the order
in which the environment schedules the processes. Then the protocol Kripke frame 
can be written as $P=\la S, \sim^A \ra$, where $S= \Gz \times  AC$ which is indeed the cartesian
product defined in Lemma \ref{lem:Kprod} between frame $\Gz$ and the action frame
$AC$. Because the protocol is defined as a product, we have the following simplification
for task solvability~: 

\begin{theorem}
\label{thm:Kripketasksolv}
In the category of Kripke frames, task solvability in the sense of Definition
\ref{def:solvingTaskK} is equivalent to the
existence of a morphism $h : P \rightarrow \Delta$ such that $\pi_I \circ h$ is equal to $\pi_I : P \rightarrow
I$, i.e. the following diagram commutes~: 
\begin{center}
\begin{tikzpicture}[scale=0.6, every node/.style={scale=0.6}]
  \node (s) {$P$};
  \node (xy) [below=2 of s] {$\Delta \subseteq I \times O$};
  \node (x) [left=of xy] {$I$};
  \draw[<-] (x) to node [sloped, above] {$\pi_I$} (s);
  \draw[->, dashed, left] (s) to node {$h$} (xy);
  \draw[->] (xy) to node [below] {$\pi_I$} (x);
\end{tikzpicture}
\end{center}
\end{theorem}

\begin{proof}
Suppose we have task solvability in the sense of Definition \ref{def:solvingTaskK}. 
Hence we have a morphism $f$ and carrier maps (identified with thick arrows below)
$\cP$ such that the following diagram commutes ``up to inclusion''~: F
\begin{center}
\begin{tikzpicture}
  \node (s) {$P$};
  \node (x) [left=of xy] {$I$};
  \node (y) [right=of xy] {$O$};
  \draw[->] (s) to node [sloped, above] {$f$} (y);
  \draw[thick,->] (x) to node [sloped, above] {$\cP$} (s);
\draw[thick,->] (x) to node [below] {$\Delta$} (y);
\end{tikzpicture}
\end{center}
First, we note, as observed after Definition \ref{def:taskK} that the carrier
map below is a composite of the inverse image of $\pi_{\Gz}$ with $\pi_{\Gf}$. Similarly,
as $P$ is a product frame, the carrier map $\cP$ is induced by the structural
map $\pi_{\Gz}: P \rightarrow \Gz$. Hence task solvability is equivalent to the
commutation of the diagram below, where $p$ is any section of the projection 
$\pi_{\Gz}: \Delta \rightarrow \Gz$, and $q$ is any section of the projection
$\pi_{\Gz}: P \rightarrow \Gz$~: 
\begin{center}
\begin{tikzpicture}
  \node (s) {$P$};
  \node (xy) [below=2 of s] {$\Delta \subseteq I \times O$};
  \node (x) [left=of xy] {$\Gz$};
  \node (y) [right=of xy] {$\Gf$};
  \draw[->] (s) to node [sloped, above] {$f$} (y);
 \draw[->] (xy) to node [sloped, above] {$\pi_{\Gf}$} (y);
\draw[->] (x) to node [sloped, above] {$p$} (xy);
\draw[bend left,<-,dashed] (x) to node [sloped, above] {$\pi_{\Gz}$} (xy);
\draw[bend left,dashed,<-] (x) to node [sloped,above] {$\pi_{\Gz}$} (s);
\draw[->] (x) to node [sloped,below] {$q$} (s);
\end{tikzpicture}
\end{center}

Note that by the universal property of cartesian products, this is equivalent
to the existence of a weak morphism $h$ from $P(I)$ to $I \times O$, which 
furthermore has to be in values in $\Delta \subseteq I\times O$, such that the
following diagram commutes~: 
\begin{center}
\begin{tikzpicture}
  \node (s) {$P$};
 \node (xy) [below=2 of s] {$\Delta \subseteq I \times O$};
  \node (x) [left=of xy] {$I$};
  \node (y) [right=of xy] {$O$};
  \draw[->] (s) to node [sloped, above] {$f$} (y);
\draw[<-] (x) to node [sloped, above] {$\pi_{\Gz}$} (s);
  \draw[->, dashed, left] (s) to node {$h$} (xy);
  \draw[->] (xy) to node [below] {$\pi_{\Gz}$} (x);
  \draw[->] (xy) to node [below] {$\pi_{\Gf}$} (y);
\end{tikzpicture}
\end{center}
Hence the result.
\end{proof}

\section{A dynamic epistemic logic semantics for distribu\-ted computing}

\label{sec:dynEpSemantics}
We consider a Kripke model that represents the possible initial states in the model,
and use dynamic epistemic logic, to represent the environment as an action model transforming the initial
model into Kripke model representing the epistemic state at the end of $N$ steps of protocol $D$.
Furthermore, we express desired epistemic states to be reached, by  specifying a task also as an action model.

\subsection{Dynamic epistemic logic}
We adhere to the notation of~\cite{DEL:2007}.
Let $\I{AP}$ be a countable set of \emph{atomic propositions} (i.e.,
\emph{propositional variables}).
The set of {\em literals} over $\I{AP}$ is
$\Lit(\I{AP})= \I{AP} \cup \{\neg p \mid p \in \I{AP}\}$.
The {\em complement} of a literal $p$ is defined by
$\overline{p} = \neg p$ and $\overline{\neg p} = p$, $\forall p \in \I{AP}$.
If $X \subseteq \Lit(\I{AP})$, then
$\overline{X}= \{ \overline\ell \mid \ell \in X \}$; 
$X$ is {\em consistent} iff 
$\forall$ $\ell \in X$, $\overline{\ell} \notin X$;
and $X$ is {\em $\I{AP}$-maximal}  iff
$\forall$ $p \in \I{AP}$, either $p \in X$ or $\neg p \in X$.

\begin{example}
\label{exap}
Consider distributed protocols that manipulate boolean values, i.e. 
where for all $i \in \{0,\ldots,n\}$, 
the local values that agents $a_i$ can take are ${L}_i=\{0, 1\}$. 
$\I{AP}$ is the set $\{l_i \ \mid \ i=0,\ldots,n\}$ of predicates whose informal
interpretation is $l_i$ is true when agent $a_i$ holds local value $1$ (otherwise,
it holds 0). 
\end{example}

\begin{definition}[syntax] 
Let $\I{AP}$ be a countable set of propositional variables and $A$ a set of agents.
The language $\mathcal{L}_K$ is generated by the following BNF
grammar:
\[
\varphi ::= p \mid \neg\varphi \mid (\varphi \land \varphi) \mid
K_a\varphi
\]
\end{definition}

\begin{definition}[Semantics of formulas]
Epistemic state $(M,s)$ with $M=\la S,\sim,V \ra$, and
$\varphi \in \Lcal_K(A,\I{AP})$.
The satisfaction relation, determining when a formula is true in an
epistemic state, is defined as:

\begin{tabular}{lrl}
$M,s \models p$ & iff & $p \in L(s)$\\
$M,s \models \neg \varphi$ & iff & $M,s \not\models \varphi$\\
$M,s \models \varphi \wedge \psi$ & iff & $M,s \models \varphi
\mbox{ and } M,s \models \psi$\\
$M,s \models K_a \varphi$ & iff & $\mbox{for all } s' \in S : s
\sim_a s' \mbox{ implies } M,s' \models \varphi$\\
\end{tabular}
\end{definition}

shorthand:
\[
E_B \varphi = \bigwedge_{b\in B} K_b \varphi
\]

Hence, 
an agent $a$ is said to know an assertion 
in a state $(M,s)$
iff that assertion of true in all the states it considers possible,
given $s$. 
Therefore, the \emph{knowledge} $K_a$ of an agent $a$ with respect to
a state $s$ is the set of formulas which are true in all states
$a$-accessible from $s$.

\begin{example}
\label{ex:kripkemodel}
Consider  boolean tasks where both input and output values to the processes are taken from the set $\{0,1\}$ as in Example \ref{exap} (and we use the same notations as
in this example).
Consider now the following Kripke model, where the squiggly arrows picture
the map $L$, associating to each state of the underlying Kripke frame, AP-maximal sets of literals~:
\begin{center}
\begin{tikzpicture}[auto,line/.style={draw,thick,-latex',shorten >=2pt},decoration=snake]
\matrix[column sep=2mm,row sep=2mm]
{
\node (alpha) {$\alpha$}; & \node (beta) {$\beta$}; &
\node (gamma) {$\gamma$}; \\
\node (Lalpha) {$\scriptscriptstyle \neg l_0,\neg l_1$}; & \node (Lbeta) {$\scriptscriptstyle l_0,\neg l_1$}; &
\node (Lgamma) {$\scriptscriptstyle l_0,l_1$}; \\
};
\draw (alpha) -- node[above] {$\scriptscriptstyle a_1$} (beta);
\draw (beta) -- node[above] {$\scriptscriptstyle a_0$} (gamma);
\draw[->,dashed,decorate] (alpha) -- (Lalpha);
\draw[->,dashed,decorate] (beta) -- (Lbeta);
\draw[->,dashed,decorate] (gamma) -- (Lgamma);
\end{tikzpicture}
\end{center}
We have now the following knowledge, known to agents~: 
\begin{itemize}
\item In $\alpha$ : $K_0 \neg l_1$ and $K_0 \neg l_0$ ($a_0$ knows both its value, which is 1, and the
value of $a_1$, which is 1) ; $\neg K_1 \neg l_0$ and $K_1 \neg l_1$ ($a_1$ does not know if the value of
$a_0$ is 1, but knows its own value, which is 1)
\item In $\beta$ : $\neg K_0 \neg l_1$ and $K_0 l_0$ ($a_0$ knows its value, which is 0, but does not know
if the value of $a_1$ is 1) ; $\neg K_1 l_0$ and $K_1 \neg l_1$ ($a_1$ does not know if the value of $a_0$
is 1 but it knows its value is 0)
\item In $\gamma$ : $\neg K_0 \neg l_1$ and $K_0 l_0$ ($a_0$ knows its value, which is 0, but does not
know if $a_1$ has value 0) ; $K_1 l_0$ and $K_1 l_1$ ($a_1$ knows both its value, which is 1, and the
value of $a_0$ which is 1)
\end{itemize}
Knowledge is sparse though : one iteration of the knowledge semantics shows that in states $\alpha$
and $\beta$, $\neg K_1 (K_0 \neg l_1)$, so that $a_1$ does not know if $a_0$ knows the value of $a_1$ in
these two states (which is 1), and in $\beta$ and $\gamma$, $\neg K_0(K_1 l_0)$, that is, $a_0$ does
not know if $a_1$ knows the value of $a_0$ in these two states (which is 0). One could guess that even
starting with this restricted set of states, no agent can decide, coherently, to reach consensus. We will come back to this in more generality in Section \ref{commonknowledge}. 
\end{example}

\begin{definition}[Language of action model logic]
\emph{
$\varphi \in \Lcal^{\RM{stat}}_{\I{KC}\otimes}(A,P)$ and
$\alpha \in \Lcal^{\RM{act}}_{\I{KC}\otimes}(A,P)$ are defined by:
\begin{align*}
\varphi & ::= p \mid (\neg\varphi) \mid (\varphi \wedge \varphi) \mid
K_a \varphi \mid [\alpha]\varphi\\
\alpha & ::= (\Msf,\ssf) 
\end{align*}
where $p \in P$, $a \in A$.
}
\end{definition}

\begin{definition}[Semantics of formulas and actions]
Epistemic state $(M,s)$ with $M=\la S,\sim,L \ra$,
action model $\Msf = \la \Ssf,\sim,\pre \ra$, and
$\varphi \in \Lcal^{\RM{stat}}_{\I{KC}\otimes}(A,P)$ and
$\alpha \in \Lcal^{\RM{act}}_{\I{KC}\otimes}(A,P)$.

\begin{tabular}{lrl}
$M,s \models p$ & iff & $p \in L(s)$\\
$M,s \models \neg \varphi$ & iff & $M,s \not\models \varphi$\\
$M,s \models \varphi \wedge \psi$ & iff & $M,s \models \varphi
\mbox{ and } M,s \models \psi$\\
$M,s \models K_a \varphi$ & iff & $\mbox{for all } s' \in S : s
\sim_a s' \mbox{ implies } M,s' \models \varphi$\\
$M,s \models [\alpha]\varphi$ & iff & $\mbox{for all } M,\Msf,s,\ssf :
M,s \models \pre(\ssf) \mbox{ implies }$\\
& & $(M \otimes \Msf,(s,\ssf)) \models \varphi$ \\
\end{tabular}


The \emph{restricted modal product} $M \otimes \Msf = \la S',\sim',L'
\ra$ is defined as:

\begin{tabular}{lrl}
$S'$ & $=$ & $\{(s,\ssf) \mid s \in S, \ssf \in \Ssf, \mbox{ and } M,s
  \models \pre(\ssf)\}$\\
$(s,\ssf) \sim_a' (t,\tsf)$ & iff & $s \sim_a t \mbox{ and } \ssf \sim_a
  \tsf$\\
$p \in L'(s,\ssf)$ & iff & $p \in L(s)$
\end{tabular}

\end{definition}

Now, we can organize Kripke models as a category, by defining Kripke model morphisms~: 

\begin{definition} [Morphism of Kripke models]
\label{def:KModMorph}
Let $M=\la S, \sim^A,L^{\I{AP}} \ra$ and $N=\la T,\sim^A,L^{\I{AP}} \ra$ be two Kripke
models. A morphism of Kripke models is a
morphism $f$ of the underlying Kripke frames of $M$ and $N$ such that 
$  L^{\I{AP}}(f(s)) \subseteq  L^{\I{AP}}(s) $ for all states $s$ in $S$. 
\end{definition}

We note $\cal KM$ for the category of Kripke models with such morphisms. We note that
$\cal KM$ is cartesian as is $\cal K$ the category of Kripke frames, with the
cartesian product being defined as the Kripke model with underlying Kripke frame
being the cartesian product of the underlying Kripke frames and the function
$L^{\I{AP}}$ is defined on the product states $(s,t)$ by $L^{\I{AP}}(s,t)=L^{\I{AP}}(s)
\cup L^{\I{AP}}(t)$. 

We prove now that morphisms can only ``lose knowledge'' (whereas it
is well-known that $p$-morphisms, \cite{sep-dynamic-epistemic} keep knowledge invariant)~: 

\begin{theorem}
\label{prop:prop1}
Consider now two Kripke models  
$M'=\la S',\sim'^A , L'\ra$ and 
$M = \la S,\sim^A, L\ra$, and a 
morphism $f$ from 
$M$ to 
$M'$. 
Then, $f$ ``can only lose knowledge'' for every agent $a$, i.e. for all $a \in A$, 
for all states $s\in M$, $M',f(s) \models K_a \phi \Rightarrow M,s \models K_a \phi$. 
\end{theorem}


\begin{proof}
Recall that~:
 
\emph{
\begin{tabular}{lrl}
$M,s \models K_a \varphi$ & iff & $\mbox{for all } s' \in S : s
\sim_a s' \mbox{ implies } M,s' \models \varphi$\\
\end{tabular}
}. 

Consider $t \sim_a s$~:  
$f(t) \sim_a f(s)$ and as $M',f(s) \models K_a \phi$, 
by definition of the semantics of $K_a$ that we recapped above, we know
that $M',f(t) \models \phi$. Therefore $\phi \in L^{AP}(f(t))$, and by definition
of morphisms of Kripke models, $\phi \in L^{AP}(f(t)) \subseteq L^{AP}(t)$. 
So $M,t \models \phi$ and $M,s \models K_a \phi$. 
\end{proof}

Note that knowledge is not  defined in the topological
approach to distributed computing~\cite{HerlihyKR:2013}.
We will  establish below mappings between simplicial
complexes and epistemic models,
enabling us to prove topological assertions about epistemic models, and knowledge assertions about models based on complexes.
The facets of simplicial complexes
correspond to the states of epistemic models.


\subsection{Action models for tasks}
\label{sec:logicTasks}

We now turn our attention to information change. Recall that
an \emph{action model} is a structure $\Msf = \mbox{$\la \Ssf,\sim,\pre \ra$}$,
where $\Ssf$ is a domain of \emph{action points}, such that for
each $a \in A$, $\sim_a$ is an equivalence relation on $\Ssf$, and
$\pre : \Ssf \to \Lcal$ is a preconditions function that assigns a
\emph{precondition} $\pre(\ssf)$ to each $\ssf \in \Ssf$.



Each action can be thought of as an announcement made by the environment,
which is not necessarily public, in the sense that
not all system agents receive these announcements.

Let $M=\la S, \sim^A,L^{\I{AP}} \ra$ be a Kripke model and $A = \la T,\sim,\pre \ra$
 be an action model. The
 \emph{product update model} $M[A]= \la S\times T, \sim^A,L^{\I{AP}} \ra$, where
each world of  $M[A]$ is a pair $(s,t)$ where $s\in S,t\in T$, such that $\pre(t)$ holds in $s$.
Then, $(s,t)\sim_a (s',t')$ whenever it holds that $s\sim_a s'$
and $t\sim_a t'$. The valuation of $p$ at a pair $(s,t)$ is just as it was at $s$.

In this definition, the worlds of $M[A]$ are obtained by making  copies of each world $s$ of $M$, one 
copy per event $t\in A$. A pair $(s,t)$ is to be included in the worlds of $M[A]$ 
if and only if $M$ satisfies the precondition $preA(t)$ of event $t$.\footnote{Usually pointed 
Kripke models and action models are used, but we do not need them here.}.

We extend the task formalism in terms of Kripke frames from Section~\ref{sec:tasksolvability} to define
an epistemic model. Recall that
a {task} is a triple $\cT=\la \Gz,\Gf,\Delta \ra$, where
the set of initial initial states $\Gz$, is such that in every initial state, the environment is
in the same state, and  $\Gf$ is a set of states, with the same environment state.
Also, $\Delta$ is a carrier morphism of Kripke
frame $I=\la \Gz, \sim^A \ra$ to Kripke frame $O=\la \Gf,\sim^A \ra$.
For the semantics, we assume that the local state of a process in either an input or an output state, 
is determined by the value of a variable. And just as an example simplify notation, we can assume these local variables
take values from the set $\{0,1\}$.
Thus, it is sufficient to consider atomic propositions $l_i$ interpreted as ``$l_i$ is true when agent $a_i$ has
value 1, otherwise it has value 0'', as in Example~\ref{ex:kripkemodel}.
Then, the \emph{input model}  is $ \la \Gz,\sim^A,L^{\I{AP}} \ra$, where
 $\la \Gz, \sim^A\ra$ is the input  Kripke frame (see Definition~\ref{def:taskK}),
  and the function
$L^{\I{AP}} : S \to 2^{\Lit(\I{AP})}$ 
 is such that $l_i$ is true in a state $s$ precisely when the input variable of  $a_i$ is $1$.
 
In this case the indistinguishability relation $\sim^A$ matches the meaning of the atomic propositions, in the
following sense. An agent $a_i$ knows the value of its input variable in a state $u$, because the atomic proposition $l_i$
has the same value in every $v$ such that $u\sim_i v$. More generally, one can infer  $\sim^A$ from the values of the atomic
propositions (and vice-versa, modulo renaming of names of atomic propositions).

The action model for the task,
$\cT = \mbox{$\la \Ssf,\sim,\pre \ra$}$, is defined as follows.
The {action points}, $\Ssf$ is identified with $\Gf$, so action point $ac=\la d_0,\ldots,d_n \ra$ is 
interpreted as ``agent $a_i$ decides value $d_i$'', with precondition that is true in every input state $u$ such
that $ac\in \Delta(u)$.

\subsection{Action models in distributed computing}

We assume that each state $s\in\Gz$ has the same environment state, $\epsilon$, representing
that the shared memory is empty. To study tasks,
 assume the  state of a process in $s$ is determined by the contents of a read-only local variable,
which can take values from some finite set of values.  
Also, we assume the protocol $D$ is {scheduler oblivious,} to make the framework more uniform, see below.

 Fix a {model of distributed computation}
 $M= (\Gz,P_e,\tau,\Psi)$,  with initial states $\Gz\subseteq\cG$, and a deterministic protocol $D$, with some integer $N$,
 and consider
the set  $\cS_N(D,M, \Gz)$  of $N$-admissible executions, where each process executes the same number
 of actions can be seen as a product of $\Gz$ and a set of composed scheduling actions $R$ as in Example~\ref{ex:RWsystems},
where each composed
 block action is of the form $sc_1\odot sc_2\odot\cdots\odot sc_R$.

 The \emph{action model for $M$} is $A = \la R,\sim,\pre \ra$, where each action in $R$ is
 a composed
 block action  of the form $sc_1\odot sc_2\odot\cdots\odot sc_R$,
 and $\pre$ is empty. That is, we allow each action to be applied to any initial state\footnote{
 A  more precise semantics, could uses a copy of each block action
 for each input state, and specify  a precondition that allows the copy to be applied only to that state;
 see discussion below on  scheduler oblivious protocols.}.
 Then, $\sim$ is obtained directly from protocol frame in Section~\ref{sec:carrierprod}.
 Recall that the protocol Kripke frame 
can be written as $P=\la S, \sim^A \ra$, where $S= \Gz \times  AC$  the cartesian
product defined in Lemma \ref{lem:Kprod} between frame $\Gz$ and the action frame
$AC$. More precisely, for two states $sc,sc' \in R$, it holds $sc\sim sc'$ whenever for two states
$u,v\in \Gz$ with $u\sim_a v$, it holds that $u\odot sc \sim_a v\odot sc'$ in $P$.

For the following lemma, we assume the protocol $D$ is \emph{scheduler oblivious,} 
in the sense that the protocol's behavior on two initial states with the same schedule is the same,
except for possible differences in output values (as affected by input values). Equivalently, we could
take the following lemma as an assumption. For   distributed task computability purposes,
the protocol $D$ can be assumed to be full-information (remember all the past, and send the whole local
state in each write operation), which is certainly scheduler oblivious. 


\begin{lemma}
\label{lemm:actIndW}
The relation $\sim$ of an action model is well defined. Namely, consider two $sc,sc' \in R$.
If for two states
$u,v\in \Gz$ with $u\sim_a v$, it holds that $u\odot sc \sim_a v\odot sc'$ in $P$,
then for any two other states $u',v'\in \Gz$ with $u'\sim_a v'$
 it also holds that $u'\odot sc \sim_a v'\odot sc'$.
\end{lemma}

An action model   $A = \la R,\sim,\pre \ra$ for a model $M$ with protocol $D$ and $N$
does indeed depend $M$, $D$ and $N$, although the action points $R$ as we defined
may not seem to do so. An action point $ac$ in $R$ is typically just a sequence of sets
of agents to be scheduled, and hence the same set $R$ can be used for different models.
Thus, remarkably, the effect of different models is captured solely by $\sim$.

We have defined an input epistemic model
$I=\la S, \sim^A,L^{\I{AP}} \ra$  and an action model $A = \la T,\sim,\pre \ra$,
and we get 
the {product update model} $I[A]= \la S\times T, \sim^A,L^{\I{AP}} \ra$,
whose Kripke frame is exactly the protocol Kripke frame of Section~\ref{lem:Kprod},
but decorated with the atomic propositions $l_i$. Moreover, a state  in the product is of the form $(u,ac)$,
where $u$ is an initial state in $\Gz$ and $ac$ is an action in $R$, and it
has the same propositions that the input state $u$.

The action model for tasks is defined in an analogous way, and there is an epistemic
model for the output Kripke frame by the corresponding epistemic product.


In the important following result, we lift the task solvability condition of Theorem \ref{thm:Kripketasksolv}
from the category of Kripke frames to the category of Kripke models~: 

\begin{theorem}
\label{thm:Kripketasksolv2}
Suppose we have input, output and specification Kripke models, respectively 
$I$, $O$ and $\Delta \subseteq I \times O$. 
Then task solvability in the sense of Definition \ref{def:solvingTaskK} is equivalent to the
existence of a morphism $h$ such that the same equation as in Theorem \ref{thm:Kripketasksolv} holds~: 
\begin{center}
\begin{tikzpicture}
  \node (s) {$P$};
  \node (xy) [below=2 of s] {$\Delta \subseteq I \times O$};
  \node (x) [left=of xy] {$I$};
  \draw[<-] (x) to node [sloped, above] {$\pi_I$} (s);
  \draw[->, dashed, left] (s) to node {$h$} (xy);
  \draw[->] (xy) to node [below] {$\pi_I$} (x);
\end{tikzpicture}
\end{center}
\end{theorem}

By Theorem \ref{prop:prop1}, we know that the knowledge of each agent (or process)
can only decrease (or stay constant) along the two (surjective) arrows $\pi_I$
of this diagram. So the diagram
above has a simple an illuminating interpretation : we can only improve knowledge from
$I$ to $P$ (the protocol should improve knowledge of all agents through 
communication). Now the task is solvable if and only if there is enough knowledge
from all agents (map $h$) to that to have at least the knowledge specified by Kripke
model $\Delta$.



\section{Combinatorial topology and Kripke models}
\label{sec:categor}

\label{sec:categories}
\comment{Kripke frames, and Kripke models can be organized as categories. The interest is that the semantics of distributed systems can be expressed using categorical operators on
Kripke frames. 
}

We now describe the equivalence of categories between (proper) Kripke
frames and (pure chromatic) simplicial complexes; the semantics of distributed systems is then
transported from Kripke models to simplicial complexes (suitably decorated). 
This shows that ``knowledge'' 
exhibits topological invariants.

\comment{
\subsection{A categorical view on Kripke frames and models}

We can define a suitable category $\cal K$ 
of Kripke frames with weak morphisms we are interested in~:
}

\comment{
\begin{definition} [$p$-morphism of Kripke frame] 
Let $M=\la S, \sim^A \ra$ and $N=\la T,\sim^A \ra$ be two Kripke frames. A p-morphism of Kripke
frame is a weak morphism $f: M \rightarrow N$ such that whenever 
$f(u) \sim_a w$ there is a $v \in W$ such that $u \sim_a v$ and $f(v)=w$. 
\end{definition}

These p-morphisms are special kinds of bisimulations of Kripke frames, i.e. generate a relation $B$
between nodes in $S$ and nodes in $T$ such that~: 
\begin{itemize}
\item if $u B u'$ and $u \sim_a v$ then there exists $v' \in T$ such that $v B v'$ and $u' B v'$
\item if $u B u'$ and $u' \sim_a v'$ then there exists $v \in W$ such that $v B v'$ and $u \sim_a v$
\end{itemize}
}

\subsection{Simplicial complexes and simplicial models}



We consider now an apparently very different structure from a Kripke graph, which has also been used to model
distributed systems \cite{HerlihyKR:2013} (for a textbook covering combinatorial topology notions see~\cite{kozlov:2007})~: 

\begin{definition} [Simplicial complex]
A simplicial complex $C$ is a family of non-empty finite subsets of a set $S$ such that
for all $X \in C$, $Y \subseteq X$ implies $Y \in C$ ($C$ is downwards closed).
\end{definition}

Elements of $S$ (identified with singletons) 
are called vertices, elements of $C$ of greater cardinality are called faces. The dimension of a face $X \in C$, $\dim X$, is the cardinality of $X$ minus one. 
The maximal faces of $C$ (i.e. faces that are not subsets of any other face) are called facets. 
The dimension of a simplicial complex is the maximal dimension of its faces. Pure simplicial complexes are simplicial complexes such that the maximal faces are all of the same
dimension. 

\begin{definition} [Simplicial maps] 
Let $C$ and $D$ be two simplicial complexes. A simplicial map $f: C \rightarrow D$ is a function that
maps the vertices of $C$ to the vertices of $D$ such that for all faces $X$ of $C$, $f(C)$ (the image set
on the subset of vertices $C$) is a face of $D$. 
\end{definition}

Now, we can define pure simplicial maps respecting facets. 
Finally, we will associate colors to each vertex of simplicial complexes, representing, as in~\cite{HerlihyKR:2013}, the names of the different processes
involved in a protocol. We also define chromatic simplicial maps as the simplicial
maps respecting colors. 

\begin{definition} [Pure simplicial complexes]
A pure simplicial complex $C$ is a simplicial complex such that all facets have the same dimension as $C$. 
\end{definition}

\begin{definition} [Pure simplicial maps] 
Let $C$ and $D$ be two pure simplicial complexes. A pure simplicial map $f : C \rightarrow D$ is
a simplicial map such that $\dim f(X)=\dim X$, for all $X \in C$, i.e. is such that $f$, restricted
to vertices of a given simplex, is injective. 
\end{definition}

\begin{definition} [Chromatic simplicial complex] 
A chromatic simplicial complex $C$ on a vertex set $S$, colored by a set of colors $L$, is a simplicial complex together
with a (coloring) map $l: S \rightarrow L$ (inducing a map, still called $l$, from $2^S$ to $2^L$) 
such that 
for all $X \in C$,
$\dim l(X)=\dim X$.
\end{definition}

The condition of the coloring map is indeed to be interpreted as 
the property that the vertices of $X$ have all distinct colors. 

\begin{definition} [Chromatic simplicial maps] 
A chromatic simplicial map is a simplicial map $f: C \rightarrow C'$ which preserves the coloring, 
i.e. such that $l'(f(X))=l(X)$, for all $X \in C$. 
\end{definition}

Note that a chromatic simplicial map is necessarily injective on each simplex, 
since as it preserves
colors, and points of any simplex are of distinct colors, $f$ cannot equate two points
of the same simplex. This also implies that for all $X$ and $Y$ chromatic simplicial complexes,
$f(X\cap Y)=f(X)\cap f(Y)$. 

Let $p{\cal CS}$ be the category of pure chromatic simplicial complexes. 

A common way to describe actual 
states in combinatorial topology is by decorating vertices of simplicial complexes
by local states of processes. 
Instead of describing actual values, we rely on a logical languages to describe the properties
of local states. 

Let $\I{AP}$ be a countable set of propositional variables and $A$ a set of agents. We use
the same notations as in Section \ref{sec:dynEpSemantics}~: 
the set of {\em literals} over $\I{AP}$ is
$\Lit(\I{AP})= \I{AP} \cup \{\neg p \mid p \in \I{AP}\}$.
If $X \subseteq \Lit(\I{AP})$, then
$\overline{X}= \{ \overline\ell \mid \ell \in X \}$; 
$X$ is {\em consistent} iff 
$\forall$ $\ell \in X$, $\overline{\ell} \notin X$;
and $X$ is {\em $\I{AP}$-maximal}  iff
$\forall$ $p \in \I{AP}$, either $p \in X$ or $\neg p \in X$.


\begin{definition}
A simplicial model is $(C, l, v)$ where 
$(C,l)$ is a pure chromatic simplicial set,  
and $v: S \rightarrow \wp(\I{AP})$ an assignment of subsets of the set of literals of $\I{AP}$  
for each state $s \in S$ such that for all facets $f=(s_0,\ldots,s_n) \in C$, 
$\bigcup\limits_{i=0}^n v(s_i)$ (that we denote as $v(f)$ by an abuse of notation) is $AP$-maximal.
\end{definition}


\comment{
\subsection{Interpretation of multi-modal S5 logics in simplicial models}

\ForAuthors{Sections to be reorganized - in particular wrt dynEpLogDC.tex}

The satisfaction relation, determining when a formula of multi-modal S5 logics is true in a
facet $f=(s_0,\ldots,s_n)$ of a simplicial model $M=(C,l,v)$, is defined as:

\begin{tabular}{lrl}
$M,f \models p$ & iff & $p \in v(f)$\\
$M,f \models \neg \varphi$ & iff & $M,f \not\models \varphi$\\
$M,f \models \varphi \wedge \psi$ & iff & $M,f \models \varphi
\mbox{ and } M,f \models \psi$\\
$M,f \models K_a \varphi$ & iff & $\varphi \in v(s_i)$ 
\end{tabular}

(where $i$, in the last statement, is the unique index such
that $l(s_i)=a$) 
}

\subsection{Equivalence between (proper) Kripke models and simplicial models}
\label{equivalencesimpcomp}

We have the main result of this section~: 


\begin{theorem}
\label{thm:equiv}
Let $A$ be a finite set and $p{\cal CS}_A$ (resp. ${\cal K}_A$) be 
the full subcategory of pure chromatic simplicial complexes with colors in $A$ (resp. the full 
subcategory of proper Kripke frames with agent set $A$).
$p{\cal CS}_A$ and ${\cal K}_A$ are equivalent categories.
\end{theorem}

\begin{proof}
We construct functors $F: p{\cal CS} \rightarrow {\cal K}$ and $G: {\cal K} \rightarrow p{\cal CS}$ as
follows. 

Let $C$ be a pure chromatic complex on the set of colors $A$. We associate $F(C)=\la S, \sim^A \ra$ 
with $S$ being the set of facets of $C$ and the equivalence relation $\sim_a$, for all $a \in A$, generated by
the relations $X \sim_a Y$ (for $X$ and $Y$ facets of $C$) if $a \in l(X \cap Y)$. 

For $f: C \rightarrow D$ a morphism in $p{\cal CS}$, consider facets $X$ and $Y$ of $C$ such 
that $X \sim_a Y$ in $F(C)$, thus $a \in l(X \cap Y)$. Because $f$ is a chromatic 
simplicial map, we have seen that $f(X \cap Y)=f(X)\cap f(Y)$. 
Because $f$ is chromatic, 
$f(X)$ is a facet of $D$ which has the same
colors as $X$, hence $a \in l(f(X))$. Similarly for $f(Y)$, which is such that $a \in l(f(Y))$. The
colors of $f(X)$ (resp. $f(Y)$) are in bijection with the vertices of $f(X)$ (resp. $f(Y)$), hence
$a \in l(f(X)\cap f(Y)$, therefore $f(X) \sim_a f(Y)$. 

Conversely, consider a Kripke frame $M=\la S,\sim^A \ra$. Define $G(M)$ to be the quotient of the coproduct of $n$-simplexes 
$\coprod\limits_{s \in S} \{v^s_0,\ldots,v^s_n\}$ where $l(v^s_i)=a_i \in A$ ($A=\{a_0,\ldots,a_n\}$) by the relation
$R$ defined as~:
$v^s_i R v^{s'}_i$ if and only if $s \sim_i s'$.
This relation is then extended on higher simplices by 
$\{x_0,\ldots,x_k\} R \{y_0,\ldots,y_l\}$ if and only if all $x_i$ are in relation
(by $R$) with some $y_{j_i}$ and inversely, all $y_j$ are in relation with some
$x_{i_j}$. $R$ can thus be seen as a sub-chromatic simplicial complex of $X \times X$ and the
quotient is well-defined, as a chromatic simplicial complex. It is pure since, as
we equate only vertices with the same colors, and that colors are in bijection with
extremal points in all simplexes, we cannot equate a simplex with a lower
dimensional simplex. The facets are the $\{v^s_0,\ldots,v^s_n\}$ where $s \in S$, since the Kripke frames
we consider being proper, we cannot equate two facets together. 

Consider now a Kripke frame $M=\la S,\sim^A \ra$ in ${\cal K}$ with agent set $A$. 
$FG(M)$ is the Kripke frame $N=\la T,\sim'^A \ra$ such that $T$ is the set of facets of
$G(M)$. But we have just seen that the set of facets of $G(M)$ are the
sets of vertices $\{v^s_0,\ldots,v^s_n\}$ (where $s \in S$), therefore, is in bijection with
$S$. Finally, in $FG(M)$, $p \sim'_a q$ if and only if $a \in l(p \cap q)$, where
$l$ is the coloring, in $G(M)$, of $p$ and $q$ which are facets in $G(M)$. But facets
in $G(M)$ are just in direct bijection with the worlds of $M$, i.e. $p=\{v^s_0,\ldots,
v^s_n\}$ and $q=\{v^t_0,\ldots,v^t_n\}$ where $s, t \in M$. Note that $l(v^s_i)=a_i$
and $l(v^t_i)=a_i$ so $a \in l(p\cap q)$ means that some $a=a_i$ for some $i$
and $v^s_i R v^t_i$. This can only be the case, by definition of $G(M)$ if $s \sim_i t$. 
This proves that $FG(M)$ and $M$ as isomorphic Kripke frame. 

Consider now a pure chromatic simplicial complex $X \in p{\cal CS}$. 
It is easily seen that $GF(X)$ is isomorphic, as a pure chromatic simplicial complex, 
to $X$, hence $p{\cal CS}$ and $\cal K$ are equivalent categories. 
\end{proof}

We  write $\bf K$
(resp. ${\bf f}: {\bf K} \rightarrow {\bf L}$) when $K$ is a Kripke frame
for the corresponding pure chromatic simplicial complex (resp. when $f$ is a
morphism of Kripke frame, for the corresponding simplicial map). 

Let $A$ be a finite set (of ``agents'') and ${\cal SM}_{A,\I{AP}}$ (resp. ${\cal KM}_{A,\I{AP}}$) be 
the full subcategory of simplicial models with colors in $A$ and vertices decorated with formulas in 
$\I{AP}$ (resp. the full 
subcategory of proper Kripke models with agent set $A$ and atomic propositions in $\I{AP}$). 
Theorem \ref{thm:equiv} extends to the following theorem, in a straightforward
manner~: 

\begin{theorem}
\label{thm:equiv2}
${\cal SM}_{A,\cG}$ and ${\cal KM}_{A,\cG}$ are equivalent categories.
\end{theorem}





\vspace{-0.1cm}

In the category of pure chromatic simplicial complexes (resp. simplicial models), 
the cartesian product of $K=(C,S,l)$ 
with $K'=(C',S',l')$ (resp. simplicial models $K=(C,S,l,v)$ and 
$K'=(C',S',l',v')$), where $C$ is a family of non-empty
finite subsets of $S$, $C'$ a family of non-empty subsets of $S'$, $l$ is a labelling map
from $S$ to $L$, and $l'$ is a labelling map from $S'$ to $L$, is defined as $K \times K'=
(D,S\times S',m)$ (resp. $K \times K'=(D,S\times S',m,w)$),
\begin{itemize}
\item $D=\{\{(s_0,s'_0),\ldots,(s_k,s'_k)\} | \{s_0,\ldots,s_k\} \in C, \{s_0,\ldots,s_k\} \in 
C', l(s_i)=l'(s'_i) \}$
\item $m(s_i,s'_i)=l(s_i)=l'(s'_i)$ (resp. $w$ is defined by $w(s)=v(s) \cup v'(s)$) 
\end{itemize}
Of course, there is a simplicial map $\pi_K$ from $K \times K'$ to $K$ and a map
$\pi_{K'}$ from $K \times K'$ to $K'$ (which are 
the first and second projection on each simplex, mapped onto subsets of simplices). And of
course, whenever we have a morphism $u$ in $p{\cal CS}$ (resp. $\cal SM$) from $U$ to $K$ and another one
$v$ from $U$ to $K'$, the map $h$ from $U$ to $K \times K'$ defined by
$h(x)=(u(x),v(x))$ is such that $u=\pi_K \circ h$ and $v=\pi_{K'} \circ h$. 
These are the direct translations under the equivalences of Theorem \ref{thm:equiv}
(resp. \ref{thm:equiv2}) of the cartesian product of Kripke frames (resp. Kripke models). 

\begin{example}
\label{ex:consensus}
We have pictured below the product ${\bf I}\times {\bf O}$, with its two projections
$\pi_{\bf O}$ and $\pi_{\bf I}$, for the input complex ${\bf I}$ for two processes, having boolean
input values, and for the output complex ${\bf O}$ for binary consensus~: 

\begin{center}
\begin{tikzpicture}[auto,line/.style={draw,thick,-latex',shorten >=2pt}]
\matrix[column sep=2mm,row sep=2mm]
{
\node (p000) {$\scriptscriptstyle (a_0,0,0)$}; & \node (p001) {$\scriptscriptstyle (a_0,0,1)$}; & & & & & & \node (p00) {$\scriptscriptstyle (a_0,0)$}; 
\\
& & \node (p100) {$\scriptscriptstyle (a_1,0,0)$};  & \node (p101) {$\scriptscriptstyle (a_1,0,1)$}; & \node (dot1) {}; 
& \node[above] (dot2) {$\pi_{\bf I}$}; & \node (dot3) {}; & & & & \node (p10) {$\scriptscriptstyle (a_1,0)$}; \\
\node (p110) {$\scriptscriptstyle (a_1,1,0)$}; & \node (p111) {$\scriptscriptstyle (a_1,1,1)$}; & & & & & & \node (p11) {$\scriptscriptstyle (a_1,1)$}; \\
& \node (dot4) {}; & \node (p010) {$\scriptscriptstyle (a_0,1,0)$}; & \node (p011) {$\scriptscriptstyle (a_0,1,1)$}; & & & & & & & \node (p01) {$\scriptscriptstyle (a_0,1)$};
\\
& \node[left] (dot5) {$\pi_{\bf O}$}; \\
& \node (dot6) {}; \\

\node (d00) {$\scriptscriptstyle (a_0,0)$}; \\ 
& & \node (d10) {$\scriptscriptstyle (a_1,0)$}; \\
& \node (d11) {$\scriptscriptstyle (a_1,1)$}; \\
& & & \node (d01) {$\scriptscriptstyle (a_0,1)$};\\
};
\draw (p110) -- (p000) -- (p100) -- (p010) -- (p110);
\draw (p101) -- (p011);
\draw [dashed] (p011) -- (p111) -- (p001);
\draw (p001) -- (p101);
\draw[->] (dot1) -- (dot3);
\draw[->] (dot4) -- (dot6);
\draw (p00) -- (p10) -- (p01) -- (p11) -- (p00);
\draw (d00) -- (d10);
\draw (d11) -- (d01);
\end{tikzpicture}
\end{center}
In the example above the vertices of the simplicial models are decorated with
predicates, as follows~: for $O$, $(a_i,j)$ is a notation for, agent $a_i$ is
such that ``its output value is $j$'' (this is the predicate that we will call $d_j$
in the sequel) ; for $I$, $(a_i,j)$ is a notation for, 
agent $a_i$ is such that ``its input value is $j$'' (this is the predicate that we
will call $l_j$ is the sequel). Thus in the product, 
$(a_i,j,k)$ means, agent $a_i$ is such that ``its input value is $j$'' and
``its output value is $k$'' (as expected for a cartesian product, this is
the set of predicates $\{l_j,d_k\}$). 
\end{example}

In fact, there is much more categorical structure in pure chromatic simplicial complexes, although
to see this, we would have to go into much more involved arguments, that we do not
need in this simplified presentation. Let us just mention that the overall categorical
structure can be derived by observing that simplicial complexes form a quasi-topos
with all (small) limits and (small) colimits \cite{Baez}, and that $p{\cal CS}$ can
be seen as a particular subcategory of a slice category of these simplicial complexes. 

\subsection{Action models, protocol complexes and task solvability}

%

Given an action model $\Msf = \mbox{$\la \Ssf,\sim,\pre \ra$}$, with 
equivalent simplicial model 
${\bf \Msf}$, and for any input Kripke
model $I$, with equivalent simplicial model $\bf I$, we form the functor
$P$ which associates the sub-Kripke frame of the cartesian product
$I\times \Msf$ (where some of the cartesian product has been filtered out due to
preconditions), as in Section \ref{sec:carrierprod}. 

\begin{example}
\label{ex:prot}
In the atomic read/write model for two processes $a_0$ and $a_1$ and 1-admissible runs, 
${\bf M}$ corresponds to the
three possible schedules of executions, with no precondition 
i.e. $P$ is the cartesian product with ${\bf \Msf}$ and
is the unique endofunctor on pure chromatic simplicial complexes, commuting with
colimits, such that
its image on a segment 
\begin{tikzpicture}[auto,line/.style={draw,thick,-latex',shorten >=2pt}]
\matrix[column sep=1.5mm,row sep=1mm]
{
\node (p00) {$\scriptscriptstyle a_0$}; & & \node (p11) {$\scriptscriptstyle a_1$}; \\
};
\draw (p00) -- (p11);
\end{tikzpicture}
is the chromatic subdivision of the segment
\begin{tikzpicture}[auto,line/.style={draw,thick,-latex',shorten >=2pt}]
\matrix[column sep=1.5mm,row sep=1mm]
{
\node (p00) {$\scriptscriptstyle a_0$}; & \node (p11) {$\scriptscriptstyle a_1$}; &
\node (p01) {$\scriptscriptstyle a_0$}; & \node (p10) {$\scriptscriptstyle a_1$}; \\
};
\draw (p00) -- (p11) -- (p01) -- (p10);
\end{tikzpicture}. In the Figure below, we show the effect of $P$
on the input complex ${\bf I}$ 
of Example \ref{ex:ex2}, together with the structure map
$\pi_I$ from $P({\bf I})$ to ${\bf I}$~: 

\begin{center}
\begin{tikzpicture}[auto,line/.style={draw,thick,-latex',shorten >=2pt}]
\matrix[column sep=1mm,row sep=1mm]
{
\node (p000u) {$\scriptscriptstyle (a_0,0)$}; & & & & & & & & & & & & \node (p00) {$\scriptscriptstyle (a_0,0)$}; \\
\node (p103l) {$\scriptscriptstyle (a_1,1)$}; & & \node (p103u) {$\scriptscriptstyle (a_1,0)$}; \\
\node (p006l) {$\scriptscriptstyle (a_0,0)$}; & & & & \node (p006u) {$\scriptscriptstyle (a_0,0)$}; \\
\node (p110l) {$\scriptscriptstyle (a_1,1)$}; & & & & & & \node (p110u) {$\scriptscriptstyle (a_1,0)$}; & & \node (dot1) {}; & \node[above] (dot2) {$\pi_I$}; & \node (dot3) {}; & & \node (p11) {$\scriptscriptstyle (a_1,1)$}; & & & & & & \node (p10) {$\scriptscriptstyle (a_1,0)$}; \\
& & \node (p003b) {$\scriptscriptstyle (a_0,1)$};  & & & & \node (p003r) {$\scriptscriptstyle (a_0,1)$}; \\
& & & & \node (p106b) {$\scriptscriptstyle (a_1,1)$};  & & \node (p106r) {$\scriptscriptstyle (a_1,0)$}; \\
& & & & & & \node (p010r) 
{$\scriptscriptstyle (a_0,1)$}; 
& & & & & & & & & & & & \node (p01) {$\scriptscriptstyle (a_0,1)$}; \\
};
\draw (p000u) -- (p103u) -- (p006u) -- (p110u);
\draw (p000u) -- (p103l) -- (p006l) -- (p110l);
\draw (p110l) -- (p003b) -- (p106b) -- (p010r);
\draw (p110u) -- (p003r) -- (p106r) -- (p010r);
\draw[->] (dot1) -- (dot3);
\draw (p00) -- (p10) -- (p01) -- (p11) -- (p00);
\end{tikzpicture}
\end{center}

\end{example}

By the equivalence of categories of Theorem \ref{thm:equiv2}, this
creates a functor that we still write $P$ associating some subcomplex
of the cartesian product ${\bf I}\times {\bf \Msf}$. Of course, since this
is a sub-object of a cartesian product, we have a canonical structure map
$\pi_I: P({\bf I}) \rightarrow {\bf I}$. 

We recall the classical approach in combinatorial topology for task solvability, see
e.g. \cite{HerlihyKR:2013}~:

\begin{definition}\label{def:decisiontask}
A decision task is a triple $T=(\mathcal{I},\mathcal{O},\Phi)$ where $\mathcal{I},\mathcal{O}$ are $n$-chromatic complexes and $\Phi$ is a map which satisfies:
\begin{enumerate}
  \item $\Phi(\sigma)$ is a subcomplex of $\mathcal{O}$.
  \item $\Phi(\tau)\subseteq\Phi(\sigma)$ if $\tau<\sigma$.
  \item $id(\sigma)=id(\Phi(\sigma))$.
\end{enumerate}
\end{definition}

In general, a map $\Phi$ from a simplicial complex $\mathcal{C}$ to a complex $\mathcal{D}$ which satisfies the conditions \emph{1} and \emph{2} of Definition~\ref{def:decisiontask} is called a \emph{carrier map}. In   Definition~\ref{def:decisiontask} each vertex represents the state of a single process. A simplex $\sigma^{(k)}$ is used to represent compatible states, of $k+1$ processes. In addition $\Phi(\sigma)$ defines which output states are legal for each input simplex $\sigma\in\mathcal{I}$. For instance in the binary consensus task each process proposes a binary value and  they have to agree in one of them. For two processes $p$ and $q$ there are four possible input configurations, which are represented as the maximal edges (maximal simplices) of $\mathcal{I}$ in Figure~\ref{fig:consensustask}. In this case if $\sigma_1=\{(p,0),(q,0)\}$ then $\Phi(\sigma_1)=\sigma_1$ or $\Phi(\sigma_2)=\mathcal{O}$ if $\sigma_2=\{(p,0),(p,1)\}$.
\begin{figure}
\begin{center}
\includegraphics[scale=0.5]{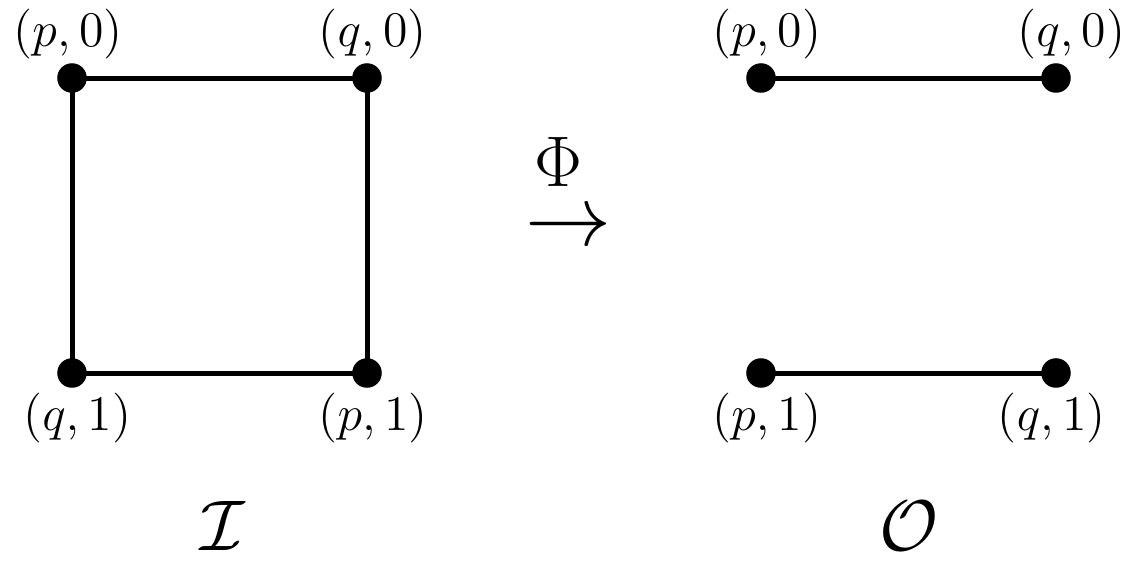}
\end{center}
\caption{The binary consensus task.}
\label{fig:consensustask}
\end{figure}

A chromatic complex $\mathcal{P}$ \emph{solves a task} $T=(\mathcal{I},\mathcal{O},\Phi)$ if there exists a simplicial map $\delta$ from $\mathcal{P}$ to $\mathcal{O}$ such that
\begin{enumerate}
  \item $\delta(\Psi(\sigma))\subseteq\Phi(\sigma)$ for all simplex $\sigma$.
  \item $id(v)=id(\delta(v))$ for all vertex $v\in\mathcal{P}$.
\end{enumerate}
where $\Psi$ is a carrier map from $\mathcal{I}$ to $\mathcal{P}$, see Figure~\ref{fig:solvetask}.
The {Asynchronous Computability Theorem} states that a decision task $T=(\mathcal{I},\mathcal{O},\Phi)$ is solvable by a wait-free protocol using read/write memory if only if there exist a chromatic subdivision $\chi^{(k)}(\mathcal{I})$ such that solves $T$.

Now, it is easy to see that, by Theorem \ref{thm:equiv}, this is exactly equivalent
to the formulation, using Kripke frames, of Section \ref{sec:tasksolvability}. 

\begin{figure}
\begin{center}
\includegraphics[scale=0.50]{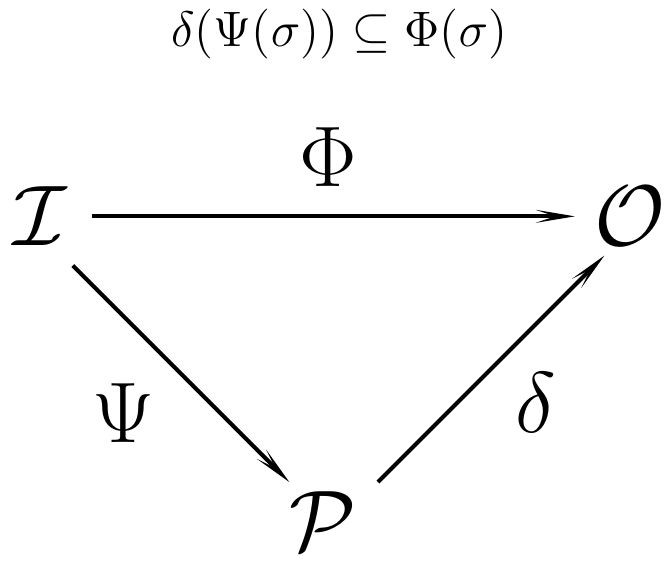}
\end{center}
\caption{Solvability of a task $(\mathcal{I},\mathcal{O},\Phi)$ by a protocol $\mathcal{P}$
with decision map $\delta$}
\label{fig:solvetask}
\end{figure}


As a matter of fact, quite generally speaking, a carrier map $\Phi$ can be seen as a functor from $p{\cal CS}$ to $p{\cal CS}$ 
that is such that $P\circ \Phi=Id_{p{\cal CS}}$. Indeed, restricting 
the latter equality on simplexes $\sigma$ of some input complex $\bf I$ implies that 
$\Phi(\sigma)$ is the subcomplex $P({\bf I})$. 
\comment{As a categorical apart\'e, this 
construction of $\Phi$, on simplexes, from $P$ extends to a functor
from pure chromatic simplicial complexes to pure chromatic simplicial complexes by
a classical categorical argument~: this is derived from 
the left Kan-extension, when it exists, 
along the Yoneda embedding, of $\Phi$ (seen as a functor from the site $\F$ of 
simplicial complexes \cite{Baez} to simplicial complexes). 
In the case in particular
where the action model $\Msf$ we are considering has empty preconditions, functor
$P$ is just the cartesian product (in the slice
category of colorings) by ${\Msf}$, which
has a right-adjoint (the $Hom$ functor, since simplicial complexes is a quasi-topos,
hence is locally cartesian-closed \cite{Baez}). This identifies carrier maps with
a left Kan-extension, which commutes with all colimits (this is then a strict
carrier map). 
}
\comment{
In the protocol complex approach of fault-tolerant distributed protocols, the specification
of the problem we want to solve is given by the specification relation ${\bf \Delta}$ from 
the input complex ${\bf I}$ to the output complex ${\bf O}$. In terms of formalisation, we will see
${\bf \Delta}$ in the what follows as a sub-(pure chromatic) simplicial complex of the product complex ${\bf I} \times {\bf O}$. 
}
We can interpret the task solvability diagram of Section \ref{sec:carrierprod} 
in purely topological terms, this interest being
that some topological invariants will prevent us from finding a map $h$ as above, showing
impossibility of the corresponding task specification. As $\pi_I$, $h$ and
$\pi_I$ are simplicial maps, we should have a corresponding diagram using any
topological invariant functor, such as homology or homotopy functors (see e.g. Example 
\ref{ex:ex2} below). In case of atomic read/write memory models, we know that
$P(I)$ corresponds to some subdivision of $I$, hence $\pi_I: P(I) \rightarrow I$ is 
a weak homotopy equivalence (which in turn implies knowledge is not improving in 
$P(I)$ with respect to $I$), this restricts a lot what task specifications
$\Delta$ can be solved, as also exemplified in Example \ref{ex:ex2}.

\begin{example}
\label{ex:ex2}
In the case of binary consensus, the task specifications is given by the Kripke
frame $\Delta$ below. The input Kripke model is composed of the four states which 
account for all binary inputs on two processes (modelled by the two predicates, 
$l_0$ for describing the input of agent $a_0$ and $l_1$ for describing the input
of agent $a_1$). The output complex is composed of the two states where the two local
values agree (we renamed the corresponding predicates for $a_i$, $d_i$)~: 
\begin{center}
\begin{tikzpicture}[auto,line/.style={draw,thick,-latex',shorten >=2pt}]
\matrix[column sep=2mm,row sep=2mm]
{
\node (alpha') {$\alpha'$}; & \node (beta') {$\beta'$}; &
\node (gamma') {$\gamma'$}; & & & & \node (alpha'') {$\alpha''$}; & \node (beta'') {$\beta''$}; & \node (gamma'') {$\gamma''$}; & \node (delta'') {$\delta''$}; \\
\node (Lalpha') {$\scriptscriptstyle \neg l_0,l_1,d_0,d_1$}; & \node (Lbeta') {$\scriptscriptstyle l_0,l_1,d_0,d_1$}; &
\node (Lgamma') {$\scriptscriptstyle l_0,\neg l_1,d_0,d_1$}; & \node (dot1) {}; & & \node (dot2) {}; & \node (Lalpha'') {$\scriptscriptstyle \neg l_0,l_1$}; & \node (Lbeta'') {$\scriptscriptstyle l_0,l_1$}; &
\node (Lgamma'') {$\scriptscriptstyle l_0,\neg l_1$}; & \node (Ldelta'') {$\scriptscriptstyle \neg l_0, \neg l_1$}; \\
\node (alpha) {$\alpha$}; & \node (beta) {$\beta$}; &
\node (gamma) {$\gamma$}; \\
\node (Lalpha) {$\scriptscriptstyle \neg l_0,l_1,\neg d_0,\neg d_1$}; & \node (Lbeta) {$\scriptscriptstyle \neg l_0,\neg l_1,\neg d_0,\neg d_1$}; &
\node (Lgamma) {$\scriptscriptstyle l_0,\neg l_1,\neg d_0,\neg d_1$}; \\
& \\
& \\
& \node (alpha0) {$\alpha 0$}; \\
& \node (Lalpha0) {$\scriptscriptstyle \neg d_0,\neg d_1$}; \\
& \node (alpha1) {$\alpha 1$}; \\
& \node (Lalpha1) {$\scriptscriptstyle d_0,d_1$}; \\
};
\draw (alpha) -- node[above] {$\scriptscriptstyle a_1$} (beta);
\draw (beta) -- node[above] {$\scriptscriptstyle a_0$} (gamma);
\draw (alpha'') -- node[above] {$\scriptscriptstyle a_1$} (beta'');
\draw (beta'') -- node[above] {$\scriptscriptstyle a_0$} (gamma'');
\draw (gamma'') -- node[above] {$\scriptscriptstyle a_1$} (delta'');
\draw (alpha''.north) to [out=30,in=150] node[above] {$\scriptscriptstyle a_0$} (delta''.north);
\draw[->,dashed,decorate] (alpha) -- (Lalpha);
\draw[->,dashed,decorate] (beta) -- (Lbeta);
\draw[->,dashed,decorate] (gamma) -- (Lgamma);
\draw (alpha') -- node[above] {$\scriptscriptstyle a_1$} (beta');
\draw (beta') -- node[above] {$\scriptscriptstyle a_0$} (gamma');
\draw[->,dashed,decorate] (alpha') -- (Lalpha');
\draw[->,dashed,decorate] (beta') -- (Lbeta');
\draw[->,dashed,decorate] (gamma') -- (Lgamma');
\draw[->,dashed,decorate] (alpha0) -- (Lalpha0);
\draw[->,dashed,decorate] (alpha1) -- (Lalpha1);
\draw[->] (Lbeta) -- node[right] {$\pi_O$} (alpha0);
\draw[->] (dot1) -- node[above] {$\pi_I$} (dot2);
\draw[->,dashed,decorate] (alpha'') -- (Lalpha'');
\draw[->,dashed,decorate] (beta'') -- (Lbeta'');
\draw[->,dashed,decorate] (gamma'') -- (Lgamma'');
\draw[->,dashed,decorate] (delta'') -- (Ldelta'');
\end{tikzpicture}
\end{center}


Using the equivalence of categories with simplicial models, we get 
${\bf \Delta} \subseteq {\bf I} \times {\bf O}$, where the state $((a_0,1,0),(a_1,1,0))$ (and
$((a_0,0,1),(a_1,0,1))$) corresponding to both processes beginning with input 1 and
deciding 0 (and corresponding to both processes beginning with input 0 and deciding 1), 
have been left out, shown below with its two projections on the input and output complexes~: 

\begin{center}
\begin{tikzpicture}[auto,line/.style={draw,thick,-latex',shorten >=2pt}]
\matrix[column sep=1mm,row sep=1.2mm]
{
\node (p000) {$\scriptscriptstyle (a_0,0,0)$}; & \node (p001) {$\scriptscriptstyle (a_0,0,1)$}; & & & & & & \node (p00) {$\scriptscriptstyle (a_0,0)$}; 
\\
& & \node (p100) {$\scriptscriptstyle (a_1,0,0)$};  & \node (p101) {$\scriptscriptstyle (a_1,0,1)$}; & \node (dot1) {}; 
& \node[above] (dot2) {$\pi_I$}; & \node (dot3) {}; & & & & \node (p10) {$\scriptscriptstyle (a_1,0)$}; \\
\node (p110) {$\scriptscriptstyle (a_1,1,0)$}; & \node (p111) {$\scriptscriptstyle (a_1,1,1)$}; & & & & & & \node (p11) {$\scriptscriptstyle (a_1,1)$}; \\
& \node (dot4) {}; & \node (p010) {$\scriptscriptstyle (a_0,1,0)$}; & \node (p011) {$\scriptscriptstyle (a_0,1,1)$}; & & & & & & & \node (p01) {$\scriptscriptstyle (a_0,1)$};
\\
& \node[left] (dot5) {$\pi_O$}; \\
& \node (dot6) {}; \\

\node (d00) {$\scriptscriptstyle (a_0,0)$}; \\
& & \node (d10) {$\scriptscriptstyle (a_1,0)$}; \\
& \node (d11) {$\scriptscriptstyle (a_1,1)$}; \\
& & & \node (d01) {$\scriptscriptstyle (a_0,1)$}; \\
};
\draw (p110) -- (p000) -- (p100) -- (p010);
\draw (p101) -- (p011);
\draw[dashed] (p011) -- (p111) -- (p001);
\draw[->] (dot1) -- (dot3);
\draw[->] (dot4) -- (dot6);
\draw (p00) -- (p10) -- (p01) -- (p11) -- (p00);
\draw (d00) -- (d10);
\draw (d11) -- (d01);
\end{tikzpicture}
\end{center}

In terms of topological invariants we see in the diagram above that the first
homology group of $P({\bf I})$ and of ${\bf I}$ are $\Z$ (they are both homotopically equivalent to the circle) and that $\pi_I$ induces the identity map in homology, 
whereas the first homology group of $\bf \Delta$ is 0. There is no factorization of the
identity through the 0 map (which is $\pi_{\bf I}$ in homology), so binary consensus is
not solvable. 

Now the pseudo consensus specification is as follows~: 

\begin{center}
\begin{tikzpicture}[auto,line/.style={draw,thick,-latex',shorten >=2pt}]
\matrix[column sep=2mm,row sep=2mm]
{
\node (p000) {$\scriptscriptstyle (a_0,0,0)$}; & \node (p001) {$\scriptscriptstyle (a_0,0,1)$}; & & & & & & \node (p00) {$\scriptscriptstyle (a_0,0)$}; 
\\
& & \node (p100) {$\scriptscriptstyle (a_1,0,0)$};  & \node (p101) {$\scriptscriptstyle (a_1,0,1)$}; & \node (dot1) {}; 
& \node[above] (dot2) {$\pi_{\bf I}$}; & \node (dot3) {}; & & & & \node (p10) {$\scriptscriptstyle (a_1,0)$}; \\
\node (p110) {$\scriptscriptstyle (a_1,1,0)$}; & \node (p111) {$\scriptscriptstyle (a_1,1,1)$}; & & & & & & \node (p11) {$\scriptscriptstyle (a_1,1)$}; \\
& \node (dot4) {}; & \node (p010) {$\scriptscriptstyle (a_0,1,0)$}; & \node (p011) {$\scriptscriptstyle (a_0,1,1)$}; & & & & & & & \node (p01) {$\scriptscriptstyle (a_0,1)$};
\\
& \node[left] (dot5) {$\pi_{\bf O}$}; \\
& \node (dot6) {}; \\

\node (d00) {$\scriptscriptstyle (a_0,0)$}; \\ 
& & \node (d10) {$\scriptscriptstyle (a_1,0)$}; \\
& \node (d11) {$\scriptscriptstyle (a_1,1)$}; \\
& & & \node (d01) {$\scriptscriptstyle (a_0,1)$};\\
};
\draw (p110) -- (p000) -- (p100) -- (p010);
\draw (p101) -- (p011);
\draw[dashed] (p011) -- (p111) -- (p001);
\draw[->] (dot1) -- (dot3);
\draw[->] (dot4) -- (dot6);
\draw (p00) -- (p10) -- (p01) -- (p11) -- (p00);
\draw (d00) -- (d10);
\draw (d11) -- (d01);
\draw (d10) -- (d01); 
\draw (p010) -- (p101);
\draw[dashed] (p000) -- (p111);
\end{tikzpicture}
\end{center}
Topologically, $P({\bf I})$, ${\bf \Delta}$ and
${\bf I}$ are homotopy equivalent
to a circle, and $\pi_I$ (the identity map in homology) can be
factored through the homology of $P({\bf I})$. Indeed, 
binary pseudo-consensus is solvable in one round in atomic read/write memory. 
\end{example}







\subsection{Common knowledge and connectivity}

\label{commonknowledge}

For $B$ a subgroup of agents (or processes), we recall, from Section \ref{sec:dynEpSemantics}
the notion of group knowledge : 
\[
E_B \varphi = \bigwedge_{b\in B} K_b \varphi
\]
Common knowledge for group $B$ is, semantically, the least solution to the equation :
\[
C_B \varphi = \varphi \wedge E_B(C_B \varphi)
\]
Given a Kripke model $M=\la S, \sim^A,L^{\I{AP}} \ra$, and a group of agents
$B=\{s_1,\ldots,s_k\}$. Define $\sim_B$ to be the transitive closure
of all the $\sim_a$, $a \in B$. This means that $s \sim_B t$ if and only if 
there exists $s_1,\ldots, s_l \in S$ such that $s \sim_{a_{i_1}} s_1 \sim_{a_{i_2}} s_2
\ldots s_l \sim_{a_{i_{l+1}}} t$ with $1 \leq a_{i_1},\ldots,a_{i_{l+1}} \leq k$. 

Now we have the following semantic characterization of $C_B \varphi$ at a state
$s$ of a Kripke model $M=\la S, \sim^A,L^{\I{AP}} \ra$~: 

$M, s \models C_B \varphi$ if and only if $M, t \models \varphi$ for all $t$ such that $s \sim_B t$.

In the sequel, by an abuse of notation, we will be identifying any Kripke model $M$ with
its simplicial model counterpart $(S,l,v)$ 
under the equivalence of Theorem \ref{thm:equiv2} between Kripke models and simplicial models,
by identifying states $S$ of Kripke models with facets of the simplicial model $(S,l,v)$. 
Given $B$ a group of agents, we form the Kripke model (and equivalently the simplicial model)
$M_B$ from $M$ by restricting the
accessibility relation to those $\sim_a$ for $a \in B$, and the formulas associated to each
state to contain only formulas describing agent in $B$. 
Then we 
have~: 

\begin{theorem}
\label{thm:commonknowledge}
The formula $\varphi$ is common knowledge at state $s$ in model $M$ if and only if $\varphi$ is true
on all facets in the same connected components of simplicial model $M$ restricted to the subgroup
$B$ of agents. 
In other words : for all facets $t$ in $M_B$ which are in the same connected component as $s$
in $M_B$, we have $\varphi \in v(t)$. 
\end{theorem}


\begin{example}
Let us consider again the binary consensus as specified in Example \ref{ex:ex2}. 
In $\Delta$, there are two connected components~: in the first one, by Theorem
\ref{thm:commonknowledge}, it is
common knowledge that $\neg l_0 \vee \neg l_1$, meaning that either $a_0$
started with $0$ or $a_1$ started with 0. In the second component, by the same
theorem, it is common knowledge that $l_0 \vee l_1$, meaning that either $a_0$
started with 1, or $a_1$ started with 1. Now, in the protocol complex of
Example \ref{ex:prot}, there is just no common knowledge involving local input
values of $a_0$ nor $a_1$. By Theorem \ref{thm:Kripketasksolv2}, if binary
consensus were solvable, we would have a morphism of Kripke model $h$ from 
the protocol Kripke model $P$ to $\Delta$, implying by Theorem \ref{prop:prop1}
that any knowledge available in $\Delta$ would be available in $P$, in particular
common knowledge. This is clearly not true here and binary consensus is not
solvable. 
\end{example}


\section{Conclusions}
\label{sec:conclusions}
{We have made a first step into defining a version of muti-agent dynamic epistemic logic
using as models higher dimensional simplical complexes. Although the step is modest,
it already shows that under a  class of action models, topological invariants 
of the initial epistemic model are fully preserved. We have worked out in detail the class of action models 
for a well studied distributed computing setting, where asynchronous processes that can crash communicate with each
other through read/write shared variables. We established a categorical equivalence between such systems
and dynamic epistemic models, and hence a precise, close relationship between the 
combinatorial topology theory of distributed computing and dynamic epistemic logic.
}

Many interesting questions are left for future work.
We have developed all our theory on facets, interpreting only the top dimensional simplexes used in distributed computing,
which is sometimes what is done, but often also simplexes of lower dimension are used to model process crashes.
Another, main avenue that we left for future study is the role of bisimulations, which are very important in dynamic epistemic
logic, and have also been considered in algebraic topology.
Of course it would be of interest to study other distributed computing settings, especially those which have stronger
communication objects available in modern multiprocessor architecture, and which are known to yield complexes
that preserve less well the topology of the input complex; indeed, their additional power is expressed as higher dimensional "holes" in
the protocol complex. It would be very interesting to find a formalization of such topological properties
in terms of knowledge, and thus obtain a generalization from common knowledge (that is tightly related to 1-dimensional connectivity)
to other form of group knowledge (related to higher-dimensional connectivity).

\subparagraph*{Acknowledgements.}
We thank Carlos Velarde and David Rosenblueth for their involvement
in the early stages of this research, and their help in developing the dual of a Kripke graph.
This work was partially supported by PAPIIT-UNAM IN109917.
Sergio Rajsbaum would like to acknowledge the Ecole Polytechnique for financial support through the 2016-2017 Visiting Scholar Program.




\begin{thebibliography}{10}

\bibitem{AADGMSsnaps}
Yehuda Afek, Hagit Attiya, Danny Dolev,  Eli  Gafni, Michael  Merritt,  and Nir Shavit.
\newblock Atomic Snapshots of Shared Memory.
\newblock {\em  J. of the ACM.}, 40(4):873--890, September 1993.

\bibitem{AttiyaR2002}
Hagit Attiya and Sergio Rajsbaum.
\newblock The combinatorial structure of wait-free solvable tasks.
\newblock {\em SIAM J. Comput.}, 31(4):1286--1313, April 2002.

\bibitem{2004dc_AW}
Hagit Attiya and Jennifer Welch.
\newblock {\em Distributed Computing: Fundamentals, Simulations, and Advanced
  Topics}.
\newblock Wiley, 2 edition, 2004.
\newblock URL:
  \url{http://www.wiley.com/WileyCDA/WileyTitle/productCd-0471453242.html}.

\bibitem{Baez}
J.~C. {Baez} and A.~E. {Hoffnung}.
\newblock {Convenient Categories of Smooth Spaces}.
\newblock {\em ArXiv e-prints}, July 2008.
\newblock \href {http://arxiv.org/abs/0807.1704} {\path{arXiv:0807.1704}}.

\bibitem{baltag&:98}
A.~Baltag, L.S. Moss, and S.~Solecki.
\newblock The logic of common knowledge, public announcements, and private
  suspiciouns.
\newblock In {\em Proceedings of the 7th conference on theoretical aspects of
  rationality and knowledge (TARK 98)}, pages 43--56, 1998.

\bibitem{BaltagM2004}
Alexandru Baltag and Lawrence~S. Moss.
\newblock Logics for epistemic programs.
\newblock {\em Synthese}, 139(2):165--224, 2004.
\newblock \href {http://dx.doi.org/10.1023/B:SYNT.0000024912.56773.5e}
  {\path{doi:10.1023/B:SYNT.0000024912.56773.5e}}.

\bibitem{sep-dynamic-epistemic}
Alexandru Baltag and Bryan Renne.
\newblock Dynamic epistemic logic.
\newblock In Edward~N. Zalta, editor, {\em The Stanford Encyclopedia of
  Philosophy}. Metaphysics Research Lab, Stanford University, winter 2016
  edition, 2016.

\bibitem{InterpretedSystems}
N. Bezhanishvili and W. Van der Hoek.
\newblock Structures for Epistemic Logics. 
\newblock In Outstanding Contributions to Logic, Volume 5, pages 175--202, 2014.

\bibitem{BiranMZ90}
Ofer Biran, Shlomo Moran, and Shmuel Zaks.
\newblock {A Combinatorial Characterization of the Distributed 1-Solvable
  Tasks}.
\newblock {\em J. Algorithms}, 11(3):420--440, 1990.

\bibitem{1993GeneralizedFLPImposibility_BG}
Elizabeth Borowsky and Eli Gafni.
\newblock Generalized flp impossibility result for {t}-resilient asynchronous
  computations.
\newblock In {\em Proc. 25th Annual ACM Symp. on Theory of Computing}, STOC,
  pages 91--100, New York, NY, USA, 1993. ACM.
\newblock URL: \url{http://doi.acm.org/10.1145/167088.167119}, \href
  {http://dx.doi.org/10.1145/167088.167119} {\path{doi:10.1145/167088.167119}}.

\bibitem{1997SimpleAlgorithmicallyReasoned_BG}
Elizabeth Borowsky and Eli Gafni.
\newblock A simple algorithmically reasoned characterization of wait-free
  computation (extended abstract).
\newblock In {\em Proceedings of the Sixteenth Annual ACM Symposium on
  Principles of Distributed Computing}, PODC '97, pages 189--198, New York, NY,
  USA, 1997. ACM.

\bibitem{BorowskyGLR01}
Elizabeth Borowsky, Eli Gafni, Nancy Lynch, and Sergio Rajsbaum.
\newblock The {BG} distributed simulation algorithm.
\newblock {\em Distributed Computing}, 14(3):127--146, 2001.

\bibitem{1993MoreChoices_Ch}
Soma Chaudhuri.
\newblock More choices allow more faults: set consensus problems in totally
  asynchronous systems.
\newblock {\em Information and Computation}, 105(1):132--158, 1993.

\bibitem{DEL:2007}
Hans~van Ditmarsch, Wiebe van~der Hoek, and Barteld Kooi.
\newblock {\em Dynamic Epistemic Logic}.
\newblock Springer Publishing Company, Incorporated, 1st edition, 2007.

\bibitem{FischerLP85}
M.~Fischer, N.~A. Lynch, and M.~S. Paterson.
\newblock {Impossibility Of Distributed Commit With One Faulty Process}.
\newblock {\em Journal of the ACM}, 32(2), April 1985.

\bibitem{GafniKM:generalizedACT:2014}
Eli Gafni, Petr Kuznetsov, and Ciprian Manolescu.
\newblock A generalized asynchronous computability theorem.
\newblock In {\em Proceedings of the 2014 ACM Symposium on Principles of
  Distributed Computing}, PODC '14, pages 222--231, New York, NY, USA, 2014.
  ACM.
\newblock URL: \url{http://doi.acm.org/10.1145/2611462.2611477}, \href
  {http://dx.doi.org/10.1145/2611462.2611477}
  {\path{doi:10.1145/2611462.2611477}}.

\bibitem{Havlicek2000}
John Havlicek.
\newblock Computable obstructions to wait-free computability.
\newblock {\em Distributed Computing}, 13(2):59--83, 2000.

\bibitem{HerlihyKR:2013}
Maurice Herlihy, Dmitry Kozlov, and Sergio Rajsbaum.
\newblock {\em Distributed Computing Through Combinatorial Topology}.
\newblock Elsevier-Morgan Kaufmann Publishers Inc., San Francisco, CA, USA, 1st
  edition, 2013.

\bibitem{HerlihyR12}
Maurice Herlihy and Sergio Rajsbaum.
\newblock Simulations and reductions for colorless tasks.
\newblock In {\em Proceedings of the 2012 ACM symposium on Principles of
  distributed computing}, PODC '12, pages 253--260, New York, NY, USA, 2012.
  ACM.

\bibitem{2013TopologuDistributedAdversaries_HR}
Maurice Herlihy and Sergio Rajsbaum.
\newblock The topology of distributed adversaries.
\newblock {\em Distributed Computing}, 26(3):173--192, 2013.

\bibitem{1999TopologicalStructureAsynchronous_HS}
Maurice Herlihy and Nir Shavit.
\newblock The topological structure of asynchronous computability.
\newblock {\em J. ACM}, 46(6):858--923, November 1999.
\newblock URL: \url{http://doi.acm.org/10.1145/331524.331529}, \href
  {http://dx.doi.org/10.1145/331524.331529} {\path{doi:10.1145/331524.331529}}.

\bibitem{HSbook}
Maurice Herlihy and Nir Shavit.
\newblock {\em The Art of Multiprocessor Programming}.
\newblock Elsevier, 1 edition, 2012.
\newblock URL:
  \url{https://www.elsevier.com/books/the-art-of-multiprocessor-programming-revised-reprint/herlihy/978-0-12-397337-5}.

\bibitem{Herlihy:waitFree1988}
Maurice~P. Herlihy.
\newblock Impossibility and universality results for wait-free synchronization.
\newblock In {\em Proceedings of the Seventh Annual ACM Symposium on Principles
  of Distributed Computing}, PODC '88, pages 276--290, New York, NY, USA, 1988.
  ACM.
\newblock URL: \url{http://doi.acm.org/10.1145/62546.62593}, \href
  {http://dx.doi.org/10.1145/62546.62593} {\path{doi:10.1145/62546.62593}}.

\bibitem{HoestS97}
Gunnar Hoest and Nir Shavit.
\newblock Towards a topological characterization of asynchronous complexity.
\newblock In {\em Proc. 16th ACM Symp. Principles of distributed computing},
  PODC, pages 199--208, New York, NY, USA, 1997. ACM.

\bibitem{Kogan:2012:MCF}
Alex Kogan and Erez Petrank.
\newblock A methodology for creating fast wait-free data structures.
\newblock In {\em Proceedings of the 17th ACM SIGPLAN Symposium on Principles
  and Practice of Parallel Programming}, PPoPP '12, pages 141--150, New York,
  NY, USA, 2012. ACM.
\newblock URL: \url{http://doi.acm.org/10.1145/2145816.2145835}, \href
  {http://dx.doi.org/10.1145/2145816.2145835}
  {\path{doi:10.1145/2145816.2145835}}.

\bibitem{kozlov:2007}
Dmitry Kozlov.
\newblock {\em Combinatorial Algebraic Topology}.
\newblock Springer, 2007.

\bibitem{LouiAA:87}
M.~C. Loui and H.~H. Abu-Amara.
\newblock {\em {Memory requirements for agreement among unreliable asynchronous
  processes}}, volume~4, pages 163--183.
\newblock JAI press, 1987.

\bibitem{LynchBook:1996}
Nancy Lynch.
\newblock {\em Distributed Algorithms}.
\newblock Morgan Kaufmann, 1996.
\newblock URL:
  \url{https://www.elsevier.com/books/distributed-algorithms/lynch/978-1-55860-348-6}.

\bibitem{MosesR2002}
Yoram Moses and Sergio Rajsbaum.
\newblock A layered analysis of consensus.
\newblock {\em SIAM J. Comput.}, 31(4):989--1021, April 2002.
\newblock URL: \url{http://dx.doi.org/10.1137/S0097539799364006}, \href
  {http://dx.doi.org/10.1137/S0097539799364006}
  {\path{doi:10.1137/S0097539799364006}}.

\bibitem{plaza:89}
J.A. Plaza.
\newblock Logics of public communications.
\newblock In M.L. Emrich, M.S. Pfeifer, M.~Hadzikadic, and Z.W. Ras, editors,
  {\em Proceedings of the 4th International Symposium on Methodologies for
  Intelligent Systems}, pages 201--216, 1989.

\bibitem{Porter}
Timothy Porter. 
\newblock Interpreted systems and Kripke models for multiagent systems from a categorical perspective.
\newblock In Theoretical Computer Science, Volume 323, Number 1, pages 235--266, 2004.

\bibitem{RajsbaumIterated2010}
Sergio Rajsbaum.
\newblock Iterated Shared Memory Models.
\newblock In {\em LATIN}, volume 6034 of {\em Lecture Notes in Computer
  Science}, pages 407--416. Springer, 2010.

\bibitem{2008IteratedRestrictedImmdediate_RRT}
Sergio Rajsbaum, Michel Raynal, and Corentin Travers.
\newblock The iterated restricted immediate snapshot model.
\newblock In {\em COCOON}, volume 5092 of {\em Lecture Notes in Computer
  Science}, pages 487--497. Springer, 2008.

\bibitem{Raynal-waitFree2005}
Michel Raynal.
\newblock Wait-free computing: An introductory lecture.
\newblock {\em Future Gener. Comput. Syst.}, 21(5):655--663, May 2005.
\newblock URL: \url{http://dx.doi.org/10.1016/j.future.2004.05.005}, \href
  {http://dx.doi.org/10.1016/j.future.2004.05.005}
  {\path{doi:10.1016/j.future.2004.05.005}}.

\bibitem{FHMVbook}
Yoram~Moses Ronald~Fagin, Joseph~Halpern and Moshe Vardi.
\newblock {\em Reasoning About Knowledge}.
\newblock MIT Press, 1 edition, 1995.
\newblock URL: \url{https://mitpress.mit.edu/books/reasoning-about-knowledge}.

\bibitem{SaksZ00}
Michael Saks and Fotios Zaharoglou.
\newblock {Wait-Free k-Set Agreement is Impossible: The Topology of Public
  Knowledge}.
\newblock {\em SIAM J. Comput.}, 29(5):1449--1483, 2000.
\newblock URL: \url{http://dx.doi.org/10.1145/167088.167122}, \href
  {http://dx.doi.org/10.1145/167088.167122} {\path{doi:10.1145/167088.167122}}.

\end{thebibliography}





\end{document}